\documentclass[11pt]{article}
\usepackage{natbib}
\usepackage{float} %Ojo, hay que cargarlo antes de hyperref que va con marsden_article, sino da problemas
\usepackage{marsden_article}
\usepackage[all]{xy}
\usepackage{amssymb}
\usepackage{enumerate,enumitem}
\usepackage{tikz-cd}
\usetikzlibrary{decorations.pathmorphing}
\usepackage[title]{appendix}
\usepackage{verbatim,xcolor,bold-extra,subcaption}
\graphicspath{ {./images/} }

\newtheoremstyle{obs}% name
  {3pt}%      Space above
  {3pt}%      Space below
  {}%         Body font
  {}%         Indent amount (empty = no indent, \parindent = para indent)
  {\bfseries}% Thm head font
  {.}%        Punctuation after thm head
  {.5em}%     Space after thm head: " " = normal interword space;
  {}%         Thm head spec (can be left empty, meaning `normal')

\theoremstyle{obs}

\newtheorem{example}[theorem]{Example}

\newcommand{\R}{\mathbb{R}}

\newcommand{\lcf}{\lbrack\! \lbrack}
\newcommand{\rcf}{\rbrack\! \rbrack}

\newcommand{\lvec}[1]{\overleftarrow{#1}}
\newcommand{\rvec}[1]{\overrightarrow{#1}}

\newcommand{\st}{\;\ifnum\currentgrouptype=16 \middle\fi|\;}

\makeindex

\begin{document}
\title{Variational order for forced Lagrangian systems}
\author{ D. Mart\'{\i}n de Diego, R. Sato Mart{\'\i}n de Almagro
\\[2mm]
{\small  Instituto de Ciencias Matem\'aticas (CSIC-UAM-UC3M-UCM)} \\
{\small C/Nicol\'as Cabrera 13-15, 28049 Madrid, Spain}}
\date{\today}

\maketitle
\begin{abstract}
We are able to derive the equations of motion for forced mechanical systems in a purely variational setting, both in the context of Lagrangian or Hamiltonian mechanics, by duplicating the variables of the system as introduced by \cite{Galley13, Galley14}. We show that this construction is useful to design high-order integrators for forced Lagrangian systems and, more importantly, we give a characterization of the order of a method applied to a forced system using the corresponding variational order of the duplicated one.
\end{abstract}

\tableofcontents

\section{Introduction}\label{sec:Intro}
In recent years, a new area of numerical analysis has been developed with the idea of designing numerical methods to integrate the evolution of ordinary and partial differential equations preserving, as much as possible, the qualitative features of the original dynamics (see \cite{hairer,serna,blanes}). This area is known as {\sl geometric integration}. 

An important property to preserve for variational Lagrangian and Hamiltonian systems is {\sl symplecticity}. A particularly elegant way to derive symplectic integrators is by discretizing Hamilton's principle, i.e. discretizing the variational principle first instead of the equations of motion. Numerical methods derived using this procedure are called {\sl variational integrators} and automatically preserve symplecticity and momentum and exhibit quasi-energy conservation for exponentially long times (see \cite{marsden-west} and references therein).

One of the difficulties is obtaining high-order methods since usually many of the typical examples derived in the literature are second order schemes, derived from the midpoint or trapezoidal rules essentially, but it is important to study the construction of high-order symplectic schemes (see \cite{hairer}). This is the objective of some recent papers in the last years (see for instance \cite{marsden-west,haleo,leoshin,campos,campos2,boma}).

A key ingredient that considerably lowers the difficulty of proving the order of an integrator is the variational error theorem stated in \cite{marsden-west}. The proof given there was correctly formulated and proved by \cite{PatrickCuell}. In consequence, the equivalent proof for forced systems given in \cite{marsden-west} was an open question. 

In this paper, apart from laying a geometric framework for Galley's theory both in Lagrangian and Hamiltonian sides, we observe that by using the duplication technique introduced by \cite{Galley13, Galley14} we can give an answer to the question of the order of forced Lagrangian systems using the proof for free Lagrangian systems given by \cite{PatrickCuell} (see also \cite{GrFe}).  

In particular, we can derive numerical methods for forced systems using the same methods that were previously developed for standard variational integrators. Preliminary numerical tests of this principle under the name of \emph{slimplectic} integrators can be found in \citep{Tsang_Galley_Stein_Turner15} and our approach can be used to rigorously derive the order of the methods previously developed for forced systems (see, for instance, \cite{OJM,PaLe} and references therein). 

\section{Forced Lagrangian and Hamiltonian systems}\label{sec:forced_systems}
A Lagrangian system is defined by a smooth manifold $Q$, the configuration manifold, and a smooth function $L: TQ\rightarrow {\mathbb R}$, the Lagrangian function, which determines the dynamics. Here $TQ$ denotes the tangent bundle of $Q$ with canonical projection $\tau_Q: TQ\rightarrow Q$. 
Let us consider local coordinates $(q^i)$ on $Q$, $1, \ldots, n=\dim Q$. The corresponding fibered coordinates on $TQ$ will be denoted by $(q^i, v^i)$ where $\tau_Q(q^i, v^i)=(q^i)$. 

Consider curves $c:[a, b]\subseteq {\mathbb R} \rightarrow Q$ of class $C^2$ connecting two fixed points $q_0, q_1\in Q$ and the collection of all these curves
\[
C^2(q_0, q_1, [a,b])=\{ c: [a, b] \rightarrow Q\, | \; c\in C^2([a, b]), c(a)=q_0, c(b)=q_1\}\; .
\]
whose tangent space is
\[
T_c C^2(q_0, q_1, [a,b])=\{ X: [a, b]\rightarrow TQ\, |\; X\in C^1([a, b]), \tau_Q\circ X=c \hbox{   and   } X(a)=X(b)=0\}\; .
\]
Given the Lagrangian function $L: TQ\rightarrow {\mathbb R}$ we define the functional: 
\[
\begin{array}{rrcl}
{\mathcal J}:& C^2(q_0, q_1, [a,b])&\longrightarrow&{\mathbb R}\\
                    & c&\longmapsto& \displaystyle{\int_a^b L(c(t), \dot{c}(t))\, \mathrm{d}t}
\end{array}
\]
\begin{definition}\label{hamp}{\bf (Hamilton's principle) }{ A curve $c\in C^2(q_0, q_1, [a, b])$ is a solution of the Lagrangian system defined by $L: TQ\rightarrow {\mathbb R}$ if and only if $c$ is a critical point of the functional ${\mathcal J}$, i.e. $d{\mathcal J}(c)=0$}
\end{definition} 
 
Using standard techniques from variational calculus, it is easy to show that the curves $c(t)=(q^i(t))$ solutions of the Lagrangian system defined by $L$ are the solutions $(c(t), \dot{c}(t))$, with $\dot{} = \frac{\mathrm{d}}{\mathrm{d}t}$, of the following system of second-order implicit  differential equations:
\[
\frac{\mathrm{d}}{\mathrm{d}t}\left( \frac{\partial L}{\partial v^i}\right)-\frac{\partial L}{\partial {q}^i}=0\; .
\]
which are the well-known Euler-Lagrange equations (see \cite{AM87}).

Let us denote by
$
S = \mathrm{d}q^i \otimes \frac{\partial}{\partial v^i}
$
and
$
\Delta = v^i \, \frac{\partial}{\partial v^i}
$
the vertical endomorphism and the Liouville vector field on $TQ$ (see \cite{ManuelMecanica} for intrinsic definitions).

The Poincar\'e-Cartan 2-form is defined by
$
\omega_L = - \mathrm{d}\theta_L \; , \, \theta_L = S^*(\mathrm{d}L)
$
and the energy function
$
E_L = \Delta(L)-L
$
which in local coordinates read as
\begin{eqnarray*}
\theta_L &=& \frac{\partial L}{\partial v^i} \, \mathrm{d} q^i \\
\omega_L &=& \mathrm{d}q^i \wedge \mathrm{d}\left( \frac{\partial L}{\partial v^i}\right)\\
E_L &=& v^i \frac{\partial L}{\partial v^i} - L (q, v)\; .
\end{eqnarray*}
 Here $S^*$ denotes the adjoint operator of $S$.

The Lagrangian $L$ is said to be regular if the Hessian matrix whose entries are
$$
W_{ij} = \frac{\partial^2 L}{\partial v^i \partial v^j}
$$
is regular, and in this case, $\omega_L$ is a symplectic form on $TQ$. We will assume in the sequel that $L$ is regular. 

The Euler-Lagrange equations are geometrically encoded as the equations for the flow of the vector field $X_{E_L}$: 
\[
\imath_{X_{E_L}}\omega_L = \mathrm{d}E_L\; .
\]

An {\sl external force} is expressed as a map $F: TQ\rightarrow T^*Q$ such that $\pi_Q \circ F=\tau_Q$ where $\tau_Q: TQ \rightarrow Q$ is the canonical projection. In coordinates $F(q^i, v^i)=(q^i, F_i(q, v))$. 
\[
\xymatrix{
 TQ \ar[rr]^{F}\ar[rd]_{\tau_Q}&& T^*Q\ar[ld]^{\pi_Q}\\\
&Q&
}
\]
where $T^*Q$ is the cotangent bundle with canonical projection $\pi_Q: T^*Q\rightarrow Q$.

Given a force we can construct a semibasic 1-form $\mu_F\in \Lambda^1(TQ)$ by $\langle \mu_F(u_q) , W_{u_q}\rangle=\langle F(u_q), T\tau_Q (W_{u_q})\rangle$, for all $W_{u_q}\in T_{u_q}TQ$. In coordinates
\[
\mu_F=F_i(q^i, v^i)\, \mathrm{d}q^i\; .
\]
Construct the vertical vector field $Z^v_F$ given by $\imath_{ Z^v_F}\omega_L=-\mu_F$. In coordinates
\[
Z^v_F=W^{ij} F_j\frac{\partial}{\partial v^i}\; ,
\]
where $W^{ij}$ denote the components of the inverse of the Hessian matrix $W$.
The dynamics of the forced Lagrangian system is determined by the vector field $X_{E_L}+Z^v_F$: 
\[
\imath_{X_{E_L}+Z^v_F}\omega_L = \mathrm{d}E_L -\mu_F\; .
\]
The integral curves of $X_{E_L}+Z^v_F$ are the solutions of the forced Euler-Lagrange equations: 
\[
\frac{\mathrm{d}}{\mathrm{d}t}\left( \frac{\partial L}{\partial v^i}\right)-\frac{\partial L}{\partial {q}^i}=F_i\; .
\]

Now, we move on to the Hamiltonian description of systems subjected to external forces. The cotangent bundle $T^*Q$ of a differentiable manifold $Q$ is equipped with a canonical exact symplectic structure $\omega_Q=d\theta_Q$, where $\theta_Q$ is the canonical 1-form on $T^*Q$ defined point-wise by
\[
(\theta_Q)_{\alpha_q}(X_{\alpha_q})=\langle \alpha_q, T_{\alpha_q}\pi_Q(X_{\alpha_q})\rangle
\]
where $X_{\alpha_q}\in T_{\alpha_q}T^*Q$, $\alpha_q\in T_q^*Q$.

In canonical bundle coordinates $(q^i, p_i)$ on $T^*Q$ the projection reads as $\pi_Q(q^i, p_i)=(q^i)$, and 
\[
\theta_Q= p_i\, \mathrm{d}q^i\; ,\ 
\omega_Q= \mathrm{d}q^i\wedge \mathrm{d}p_i\; .
\]

Given a Hamiltonian function $H: T^*Q\rightarrow {\mathbb R}$ we define the Hamiltonian vector field
\[
\imath_{X_H}\omega_Q=\mathrm{d}H\; 
\]
which can be written locally as
\[
X_H=\frac{\partial H}{\partial p_i}\frac{\partial}{\partial q^i}- \frac{\partial H}{\partial q^i}\frac{\partial}{\partial p_i}\; .
\]
 Its integral curves are determined by Hamilton's equations: 
\begin{eqnarray*}
\frac{\mathrm{d}q^i}{\mathrm{d}t}&=&\frac{\partial H}{\partial p_i}\; ,\\
\frac{\mathrm{d}p_i}{\mathrm{d}t}&=&-\frac{\partial H}{\partial q^i}\; .
\end{eqnarray*}

Let $TQ$ be the tangent bundle of $Q$ with canonical projection $\tau_Q: TQ\rightarrow Q$ and canonical coordinates $(q^i, v^i)$ where $\tau_Q(q^i, v^i)=(q^i)$. 

Given a function $H: T^*Q\rightarrow {\mathbb R}$, a Hamiltonian function, we may construct the transformation $\mathbb{F}H: T^*Q\rightarrow TQ$ where $\langle \mathbb{F}H(\alpha_q), \beta_q\rangle={\frac{\mathrm{d}}{\mathrm{d}t}\big |_{ t=0}H(\alpha_q+t\beta_q)}$.
In coordinates, $\mathbb{F}H(q^i, p_i)=(q^i, \frac{\partial H}{\partial p_i}(q,p))$. 
We say that the Hamiltonian is regular if $\mathbb{F}H$ is a local diffeomorphism, which in local coordinates is equivalent to the regularity of the Hessian matrix whose entries are:
\begin{equation*}
M^{i j} = \frac{\partial^2 H}{\partial p_i\partial p_j}.
\end{equation*}
Consider now the external forcing previously defined in the Lagrangian description and denote $F^H = \mathbb{F}H^* F: T^*Q\rightarrow T^*Q$. 
\[
\xymatrix{
 T^*Q \ar[rr]^{F^H}\ar[rd]_{\pi_Q}&& T^*Q\ar[ld]^{\pi_Q}\\\
&Q&
}
\]
It is possible to modify Hamiltonian vector field $X_H$ to obtain the forced Hamilton's equations as the integral curves of the vector field $X_H+Y^v_F$ where the vector field
$Y^v_F\in {\mathfrak X}(T^*Q)$ is defined by
\[
Y^v_F(\alpha_q)=\frac{\mathrm{d}}{\mathrm{d}t}\Big|_{t=0}(\alpha_q+t F^H(\alpha_q))\; .
\]
We will say the the forced Hamiltonian system is determined by the pair $(H, F^H)$.

Locally, 
\[
Y^v_F=F_i\left(q^j, \frac{\partial H}{\partial p_j}(q, p)\right)\frac{\partial}{\partial p_i}=F^H_i(q, p)\frac{\partial}{\partial p_i}
\]
modifying Hamilton's equations as follows:
 \begin{eqnarray}\label{qwe}
\frac{\mathrm{d}q^i}{\mathrm{d}t}&=&\frac{\partial H}{\partial p_i}(q,p)\; ,\\
\frac{\mathrm{d}p_i}{\mathrm{d}t}&=&-\frac{\partial H}{\partial q^i}(q,p)+F^H_i(q,p)\; .
\end{eqnarray}

\section{Forced Hamiltonian dynamics as free dynamics obtained by duplication}  
In this section we define a new unforced Hamiltonian system whose dynamical equations are related with the forced system $(H,F^H)$. In a final appendix we will expand on the interesting geometry (that of a cotangent groupoid) that arises with the duplication of variables. 

The idea is to duplicate variables in such a way that system (\ref{qwe}) transforms into a different free Hamiltonian system. This opens up the opportunity to analyze the original system using standard Hamiltonian techniques. 

Consider the Cartesian product $Q \times Q$ and its cotangent bundle $T^*(Q\times Q) \equiv T^*Q \times T^*Q$ with canonical projections $\mathrm{pr}_1: T^*Q \times T^*Q\rightarrow T^*Q$ and $\mathrm{pr}_2: T^*Q\times T^*Q\rightarrow T^*Q$ given respectively by $\mathrm{pr}_1(\alpha_q, \beta_{q'})=\alpha_q$ and $\mathrm{pr}_2(\alpha_q, \beta_{q'})=\beta_{q'}$ where $\alpha_q\in T^*_qQ$ and $\beta_{q'}\in T^*Q_{q'}$. Also in this space we have the \emph{inversion map} $\iota: T^*Q \times T^*Q \to T^*Q \times T^*Q$ (cf. section \ref{sec:geometry}), which in local coordinates acts as $\iota(\alpha_q, \beta_{q'}) = (\beta_{q'}, \alpha_q)$.

Let us endow $T^*Q \times T^*Q$ with the following symplectic structure:
\[
\Omega_{Q\times Q} =\mathrm{pr}_2^*\omega_Q-\mathrm{pr}_1^*\omega_Q\; .
\]
Observe that the mapping 
\[
\begin{array}{rrcc}
\Psi:& T^*Q\times T^*Q&\longrightarrow& T^*(Q\times Q)\\
     & (\alpha_q, \beta_{q'})&\longmapsto&  (-\alpha_q, \beta_{q'})
\end{array}
\]
is a symplectomorphism, i.e. $\Psi^*(\Omega_{Q\times Q})=\omega_{Q\times Q}$ where $\omega_{Q\times Q}$ is the standard symplectic form on $T^*(Q\times Q)$.

Consider the Hamiltonian $\mathbf{H}: T^*Q\times T^*Q\rightarrow {\mathbb R}$ defined by
\[
\mathbf{H}(\alpha_q, \beta_{q'}) = \left(H \circ \mathrm{pr}_2 - H \circ \mathrm{pr}_1\right)(\alpha_q, \beta_{q'}) = H(\beta_{q'})- H(\alpha_q)\, ,
\] 
satisfying $\mathbf{H} \circ \iota = - \mathbf{H}$, and the {\sl identity map} (cf. section \ref{sec:geometry})
\[
\begin{array}{rrcc}
\epsilon:& T^*Q &\rightarrow&T^*Q\times T^*Q\\
              &\alpha_q&\longmapsto& (\alpha_q, \alpha_q)
\end{array}
\]
Observe that the spaces $T^*Q$ and $\epsilon(T^*Q)$ are canonically diffeomorphic. We have the following 
 
\begin{lemma}\label{lemma-1}
The Hamiltonian vector field $X_\mathbf{H}$ given by
\[
\imath_{X_\mathbf{H}}\Omega_{Q\times Q}=\mathrm{d}\mathbf{H}
\]
 verifies that:
\begin{enumerate}[label=\roman*)]
\item \label{itm:lemma-1-itm1} ${X_\mathbf{H}}$ is tangent to $\epsilon(T^*Q)$;
\item \label{itm:lemma-1-itm2} $\left.X_\mathbf{H}\right\vert_{\epsilon(T^*Q)}=\epsilon_* (X_H)$.
\end{enumerate}
\end{lemma}
\begin{proof}
See Proposition \ref{prop:tangency_to_identities}.
\end{proof}
 
Observe that the proof of Lemma \ref{lemma-1} is quite straightforward using coordinates. In fact, if we take adapted coordinates $(q^i, p_i; Q^i, P_i)$ on $T^*Q \times T^* Q$ then 
     \[
     \mathbf{H}(q^i,  p_i; Q^i, P_i) = H(Q^i, P_i)- H(q^i, p_i)\; ,
     \]                  
and so                         
  \[
  X_\mathbf{H}= \frac{\partial H}{\partial p_i}(q,p)\frac{\partial}{\partial q^i}- \frac{\partial H}{\partial q^i}(q,p)\frac{\partial}{\partial p_i}+  \frac{\partial H}{\partial P_i}(Q,P)\frac{\partial}{\partial Q^i}- \frac{\partial H}{\partial Q^i}(Q,P)\frac{\partial}{\partial P_i}\; .                   
\]
which is obviously tangent to $\epsilon(T^*Q)$, this last space is locally given by the vanishing of the $2n$-constraints $Q^i-q^i=0$ and $P_i-p_i=0$ and moreover \ref{lemma-1}\ref{itm:lemma-1-itm2} follows immediately since $\epsilon (q^i, p_i)=(q^i, p_i; q^i, p_i)$.

 Define  $\mathbb{F}\mathbf{H}^{\times}(\alpha_q, \beta_{q'})=(-\mathbb{F}H(\alpha_q), \mathbb{F}H(\beta_{q'}))$ with Hessian matrix: 
  \[
\left(
\begin{array}{cc}
\displaystyle-\frac{\partial^2 H}{\partial p_i\partial p_j} (q,p)&0_{n\times n}\\
0_{n\times n}&\displaystyle \frac{\partial^2 H}{\partial P_i\partial P_j} (Q,P)
\end{array}
\right)\; .
\]

The following lemma is trivial but it will be interesting for us when going to the Lagrangian side
\begin{lemma}\label{lemma-2}
The transformation $\mathbb{F}\mathbf{H}^{\times}: T^*Q\times T^*Q\rightarrow TQ\times TQ$ is a local diffeomorphism if and only if $\mathbb{F}H: T^*Q\rightarrow TQ$ is a local diffeomorphism \end{lemma}

Given a Hamiltonian $\mathbf{H}$, we may want to add a function $K: T^*Q\times T^*Q \rightarrow {\mathbb R}$. In light of the result of lemma \ref{lemma-1}, if this function has the property $K(\beta_{q'}, \alpha_q) = -K(\alpha_q, \beta_{q'})$, then $\mathbf{H}_{K} = \mathbf{H} + K$ will still preserve that property and the trajectories of the resulting dynamics at the identities will remain bound to the identities. This allows a richer behaviour of the original system in $T^*Q$, and for this, functions $K$ satisfying said property are called \emph{generalized potentials} \cite{Galley13}. 

The previous results are only a preparation for our real objective, to find a purely Hamiltonian representation of systems with forces using this duplication of variables. For that, we need to introduce the definition of retraction (see \cite{Absil}): 

\begin{definition} 
A {\bf retraction} on a manifold $Q$ is a smooth mapping $R: TQ\rightarrow Q$ satisfying the following properties. If we denote by $R_q=\left.R\right\vert_{T_qQ}$ the restriction of $R$ to $T_qQ$ then: 
\begin{itemize}
\item $R_q(0_q)=q$, 
\item Identifying $T_{0_q} T_q Q \equiv  T_q Q$ then 
$T_{0_q}R_q =\hbox{id}_{T_qQ}$.
\end{itemize}
\end{definition}

In what follows, it will be interesting to introduce $\sigma: TQ\rightarrow Q\times Q$ defined by
$\sigma(v_q)=(q, R_q(v_q))$. From the previous definition it is easy to show that $\sigma$ is invertible in some neighborhood of $0_q\in T_qQ$ for any $q\in Q$. Denote this local inverse by $\tau: Q\times Q\rightarrow TQ$ which in coordinates will take the form
\[
\tau(q^i, Q^i)=(q^i, \tau^i (q,Q))\; .
\]

\begin{lemma}\label{lemma4}
Consider a chart $(U, \varphi)$ around a point $q\in Q$ then 
the map
$\sigma: TU\rightarrow Q\times Q$ defined by $\sigma(v_q)=(q, R_q(v_q))$ where $q\in U$ verifies that the map
$T_{0_q}\sigma$ in coordinates $(TU, T\varphi)$ and $(U\times U, \varphi\times \varphi)$ is 
\[
\left(
\begin{array}{cc}
I&0\\
I&I
\end{array}
\right)
\]
and, in consequence, the map $T_{(q,q)}\tau$ is represented by the matrix
\[
\left(
\begin{array}{cc}
I&0\\
-I&I
\end{array}
\right)
\]
\end{lemma}
\begin{proof}
We have that in the chosen coordinates $\sigma(q^i, v^i)=(F^i(q, v), G^i(q,v))=(q, R_q(v))$. Therefore, in coordinates, 
\[
\frac{\partial F^i}{\partial q^j}(q, 0)=\delta^i_j, \quad \frac{\partial F^i}{\partial v^j}(q, 0)=0,
\]
where $\delta^i_j$ is the Kronecker delta, and observe that
\[
\displaystyle \frac{\partial G^i}{\partial q^j}(q, 0)=\delta^i_j
\]
since $R_{q'}(0)=q'$ for $q'\in U$ and 
\[
\displaystyle \frac{\partial G^i}{\partial v^j}(q, 0)=\delta^i_j
\]
since $T_{0_q}R_q =\hbox{id}_{T_qQ}$.
\end{proof}

Typically, we can induce this kind of mappings using an auxiliary Riemannian metric $g$ on $Q$ with associated geodesic spray $\Gamma_g$ (see \cite{Carmo}). The associated exponential for a small enough neighborhood $U\subset TQ$ of $0_q$.
\[
\begin{array}{rrcl}
\hbox{exp}^{\Gamma_g}:&U\subset TQ& \longrightarrow& Q\times Q\\
                                            &v_q&\longmapsto& (q, \gamma_{v_q}(1))
\end{array}
\]
where $t\rightarrow \gamma_{v_q}(t)$ is the unique geodesic such that $\gamma'_{v_q}(0)=v_q$. 
For instance when $Q={\mathbb R}^n$ and we take the Euclidean metric, we induce the map
\[
\tau(q, q')=\left(q, q'-q\right)\quad \hbox{and}\quad \sigma(q, v)=(q, q + v)\; .
\]

Given a forced Hamiltonian system $(H, F^H)$ we can construct the function $K_F: T^*Q\times T^*Q\rightarrow {\mathbb R}$ as: 
\[
K_F(\alpha_q, \beta_{q'})= \frac{1}{2}\langle F^H(\beta_{q'}), \tau(q', q)\rangle - \frac{1}{2}\langle F^H(\alpha_q), \tau(q, q')\rangle
\]
Observe that this function satisfies the important property $K_F(\beta_{q'}, \alpha_q)=-K_F(\alpha_q, \beta_{q'})$, and thus is a generalized potential.
Consider the Hamiltonian $\mathbf{H}_{K_F}: T^*Q\times T^*Q\rightarrow {\mathbb R}$ defined by
 \[
 \mathbf{H}_{K_F}(\alpha_q, \beta_{q'}) = \mathbf{H}(\alpha_q, \beta_{q'}) + {K_F}(\alpha_q, \beta_{q'})\, ,
 \] 
which also satisfies $\mathbf{H}_{K_F}(\beta_{q'}, \alpha_q) = -\mathbf{H}_{K_F}(\alpha_q, \beta_{q'})$ by construction.
\begin{theorem}\label{theorem-1}
 The Hamiltonian vector field $X_{\mathbf{H}_{K_F}}$ given by
 \[
 \imath_{X_{\mathbf{H}_{K_F}}}\Omega_{Q\times Q}=\mathrm{d}\mathbf{H}_{K_F}
 \]
 verifies that:
\begin{enumerate}[label=\roman*)]
\item \label{itm:theorem-1-itm1} $X_{\mathbf{H}_{K_F}}$ is tangent to $\epsilon(T^*Q)$;
\item \label{itm:theorem-1-itm2} $\left.X_{\mathbf{H}_{K_F}}\right\vert_{\epsilon(T^*Q)}=\epsilon_* (X_H+Y^v_F)$.
\end{enumerate}
 \end{theorem}
 \begin{proof}
Part \ref{itm:theorem-1-itm1} is again a direct consequence of proposition \ref{prop:tangency_to_identities}.

To deduce part \ref{itm:theorem-1-itm2} observe that 
 \begin{eqnarray*}
  X_{\mathbf{H}_{K_F}}&=& \left(\frac{\partial H}{\partial p_i}(q,p)+\frac{1}{2}\frac{\partial F_j^H}{\partial p_i}(q,p)\tau^j(q, Q) \right)\frac{\partial}{\partial q^i}\\
  &&-\left( \frac{\partial H}{\partial q^i}(q,p)+\frac{1}{2}\frac{\partial F_j^H}{\partial q_i}(q,p)\tau^j(q, Q)\right.\\
  &&\left.+\frac{1}{2} F_j^H(q,p)\frac{\partial \tau^j}{\partial q^i}(q, Q)-\frac{1}{2} F_j^H(Q,P)\frac{\partial \tau^j}{\partial q^i} (Q, q) \right)
  \frac{\partial}{\partial p_i}\\
  &&+ \left(\frac{\partial H}{\partial P_i}(Q,P)+\frac{1}{2}\frac{\partial F_j^H}{\partial P_i}(Q,P)\tau^j(Q, q) \right)\frac{\partial}{\partial Q^i}\\
  &&-\left( \frac{\partial H}{\partial Q^i}(Q,P)-\frac{1}{2} F_j^H(q,p)\frac{\partial \tau^j}{\partial Q^i}(q, Q)\right.\\
  &&\left.+\frac{1}{2}\frac{\partial F_j^H}{\partial Q_i}(Q,P)\tau^j(Q, q)+\frac{1}{2} F_j^H(Q,P)\frac{\partial \tau^j}{\partial Q^i} (Q, q) \right)
  \frac{\partial}{\partial P_i}\\            
\end{eqnarray*}

Now using Lemma \ref{lemma4} we have that along the identities $\epsilon(T^*Q)$
 \begin{eqnarray*}
  X_{\mathbf{H}_{K_F}}&=& \frac{\partial H}{\partial p_i}(q,p) \frac{\partial}{\partial q^i}
  -\left( \frac{\partial H}{\partial q^i}(q,p)
   -F_j^H(q,p) \right)
  \frac{\partial}{\partial p_i}\\
  &&+ \frac{\partial H}{\partial p_i}(q,p) \frac{\partial}{\partial Q^i}
  -\left( \frac{\partial H}{\partial q^i}(q,p)
  - F_j^H(q,p) \right)
  \frac{\partial}{\partial P_i}\\            
\end{eqnarray*}
and thus $\left.X_{\mathbf{H}_{K_F}}\right\vert_{\epsilon(T^*Q)}=\epsilon_* (X_H+Y^v_F)$ as we wanted to prove. 
 \end{proof}

 Define the mapping $\mathbb{F}\mathbf{H}^{\times}_{K_F}: T^*Q\times T^*Q\rightarrow TQ\times TQ$ given in local coordinates as
\[
\mathbb{F}\mathbf{H}^{\times}_{K_F}(q^i, p_i, Q^i, P_i)=\left(q^i, -\frac{\partial \mathbf{H}_{K_F}}{\partial p_i},  Q^i, \frac{\partial \mathbf{H}_{K_F}}{\partial P_i}\right).
\]

\begin{proposition}\label{proposition-2}
If $H$ is regular then the transformation $\mathbb{F}\mathbf{H}^{\times}_{K_F}: T^*Q\times T^*Q\rightarrow TQ\times TQ$ is a local diffeomorphism in a neighborhood of $\epsilon(T^*Q)$. 
 \end{proposition}
  \begin{proof}

Locally if we take coordinates $(q^i, p_i, Q^i, P_i)$ then from the definition of ${K_F}$ we observe that 
\[
\left.\left( 
\begin{array}{c}
\displaystyle
\frac{\partial^2 {K_F}}{\partial p_i\partial p_j}
\end{array}
\right)\right\vert_{{\epsilon} (T^*Q)}=\left.\left( 
\begin{array}{c}
\displaystyle
\frac{\partial^2 {K_F}}{\partial P_i\partial P_j}
\end{array}
\right)\right\vert_{{\epsilon} (T^*Q)}=\left.\left( 
\begin{array}{c}
\displaystyle\frac{\partial^2 {K_F}}{\partial p_i\partial P_j}
\end{array}
\right)\right\vert_{{\epsilon} (T^*Q)}
=0
\]
therefore from the regularity of $H$ it is trivial to derive the regularity of $\mathbf{H}_K$ on a tubular neighborhood of ${\epsilon}(T^*Q)$.
\end{proof}

%In this case we define 
%\[
%\mathbf{L}_K(\mathbb{F}{\mathcal H}_K(\alpha_q, \beta_{q'})))=\langle (\alpha_q, \beta_{q'}), \mathbb{F}{\mathcal H}_K(\alpha_q, \beta_{q'})\rangle- H_K(\alpha_q, \beta_{q'}) ; .
%\]

\section{Forced Lagrangian dynamics as free dynamics obtained by duplication}
Now we will define a new free Lagrangian system whose dynamical equations are related with the forced system $(L,F)$. This is precisely the path chosen by \cite{Galley13} although we will not enter into details as to how and why he arrived at his formulation.

Consider the Cartesian product $Q \times Q$ and its tangent bundle $T(Q\times Q) \equiv TQ \times TQ$ with canonical projections $\widetilde{\mathrm{pr}}_1: TQ\times TQ\rightarrow TQ$ and $\widetilde{\mathrm{pr}}_2: TQ\times TQ\rightarrow TQ$. In local coordinates we have $\widetilde{\mathrm{pr}}_1(v_q, V_Q) = v_q$ and $\widetilde{\mathrm{pr}}_2(u_q, V_Q) = V_Q$. Consider also the maps $\tilde{\iota}: TQ \times TQ \rightarrow TQ \times TQ$ and $\tilde{\epsilon}: TQ \rightarrow TQ \times TQ$ locally defined as $\tilde{\iota}(v_q,V_Q)=(V_Q, v_q)$ and $\tilde{\epsilon}(v_q)=(v_q, v_q)$, respectively. In this new bundle $TQ \times TQ \equiv T(Q \times Q)$ we also have a vertical endomorphism and a Liouville field, whose local presentation in adapted coordinates $(q^i,v^i; Q^i,V^i)$ would be $S = \mathrm{d}q^i \otimes \frac{\partial}{\partial v^i} + \mathrm{d}Q^i \otimes \frac{\partial}{\partial V^i}$, $\triangle = v^i \partial_{v^i} + V^i \partial_{V^i}$.

Define a new Lagrangian $\mathbf{L}: TQ\times TQ\rightarrow \mathbb{R}$ as:
\[
\mathbf{L}(v_q, V_Q) = \left(L \circ \widetilde{\mathrm{pr}}_2 - L \circ \widetilde{\mathrm{pr}}_1\right)(v_q, V_Q) = L(V_Q) - L(v_q)\, ,
\]
Much like $\mathbf{H}$, this new Lagrangian satisfies that $\mathbf{L} \circ \tilde{\iota} = - \mathbf{L}$.

From this new Lagrangian we can define corresponding Poincar{\'e}-Cartan forms on $TQ \times TQ$, $\theta_{\mathbf{L}}$ and $\omega_{\mathbf{L}}$, as in Section 2. We can also define a new fibre derivative $\mathbb{F}\mathbf{L}^{\times}(v_q, V_Q)=(-\mathbb{F}L(v_q), \mathbb{F}L(V_Q))$. This leads us to state an analogue of lemma \ref{lemma-2}:
\begin{lemma}\label{lemma-5}
The transformation $\mathbb{F}\mathbf{L}^{\times}: TQ \times TQ \rightarrow T^*Q \times T^*Q$ is a local diffeomorphism if and only if $\mathbb{F}L: TQ\rightarrow T^*Q$ is a local diffeomorphism.
\end{lemma}

It is easy to check that with this definition of fibre derivative the following diagram commutes:
\[
\xymatrix{
TQ \times TQ \ar[rr]^{\mathbb{F}\mathbf{L}^{\times}} && T^*Q \times T^*Q\\\
TQ \ar[u]^{\tilde{\epsilon}} \ar[rr]^{\mathbb{F}L} && T^*Q \ar[u]^{\epsilon}
}
\]
We can also define the energy of the system as in the usual case, with $E_{\mathbf{L}} = \triangle\left(\mathbf{L}\right) - \mathbf{L}$, but in order to relate this with the Hamiltonian formulation, it will be useful to rewrite it as:
\begin{equation*}
E_\mathbf{L}(v_q, V_Q) = \left\langle \mathbb{F}\mathbf{L}^{\times}(v_q, V_Q), (v_q, V_Q)\right\rangle_{\times} - \mathbf{L}(v_q, V_Q).
\end{equation*}
where $\left\langle \cdot,\cdot \right\rangle_{\times}: (T^*Q \times T^*Q) \times (TQ \times TQ) \to \mathbb{R}$ is the inner product defined as:
\begin{equation}
\left\langle \alpha , v \right\rangle_{\times} = \left\langle \mathrm{pr}_2(\alpha), \widetilde{\mathrm{pr}}_2(v)\right\rangle - \left\langle \mathrm{pr}_1(\alpha), \widetilde{\mathrm{pr}}_1(v)\right\rangle.
\end{equation}

The following results will help us prove the analogue of lemma \ref{lemma-1}.
\begin{proposition}
\label{prop:inversion_commutation}
Let $\mathbf{L}: TQ \times TQ \to \mathbb{R}$ be such that $\mathbf{L} \circ \tilde{\iota} = - \mathbf{L}$ then
the following diagram commutes:
\[
\xymatrix{
TQ \times TQ \ar[d]^{\tilde{\iota}} \ar[r]^{\mathbb{F}\mathbf{L}^{\times}} & T^*Q \times T^*Q \ar[d]^{\iota}\\\
TQ \times TQ \ar[r]^{\mathbb{F}\mathbf{L}^{\times}} & T^*Q \times T^*Q
}
\]
\end{proposition}

\begin{proof} We need to show that:
\begin{equation*}
\mathbb{F}\mathbf{L}^{\times} \circ \tilde{\iota} = \iota \circ \mathbb{F}\mathbf{L}^{\times}.
\end{equation*}
On the left-hand side we have
\begin{align*}
\left(\mathbb{F}\mathbf{L}^{\times} \circ \tilde{\iota}\right)(q,v,Q,V) &= \mathbb{F}\mathbf{L}^{\times} (Q,V,q,v)\\
&= \left(Q, -\frac{\partial \mathbf{L}}{\partial V}(Q,V,q,v),q,\frac{\partial \mathbf{L}}{\partial v}(Q,V,q,v)\right),
\end{align*}
while on the right-hand side we have
\begin{align*}
\left(\iota \circ \mathbb{F}\mathbf{L}^{\times}\right)(q,v,Q,V) &= \iota\left(q,-\frac{\partial \mathbf{L}}{\partial v}(q,v,Q,V),Q,\frac{\partial \mathbf{L}}{\partial V}(q,v,Q,V)\right)\\
&= \left(Q,\frac{\partial \mathbf{L}}{\partial V}(q,v,Q,V),q,-\frac{\partial \mathbf{L}}{\partial v}(q,v,Q,V)\right)
\end{align*}

Now, using that $\mathbf{L} \circ \tilde{\iota} = - \mathbf{L}$ we find:
\begin{align*}
\frac{\partial \mathbf{L}}{\partial V}(Q,V,q,v) &= - \frac{\partial \mathbf{L}}{\partial V}(q,v,Q,V),\\
\frac{\partial \mathbf{L}}{\partial v}(Q,V,q,v) &= - \frac{\partial \mathbf{L}}{\partial v}(q,v,Q,V).
\end{align*}

Applying this we immediately arrive at the desired result.
\end{proof}

\begin{lemma}
\label{lem:inner_prod_inversion}
The inner product $\left\langle\cdot,\cdot\right\rangle_\times$ satisfies that
\begin{equation*}
\left\langle \iota(\alpha),\tilde{\iota}(v)\right\rangle_\times = - \left\langle\alpha,v\right\rangle_\times.
\end{equation*}
\end{lemma}

\begin{proof}
First note that $\mathrm{pr}_1 \circ \iota = \mathrm{pr}_2$ and $\widetilde{\mathrm{pr}}_1 \circ \tilde{\iota} = \widetilde{\mathrm{pr}}_2$, and that the same holds under the exchange $1 \leftrightarrow 2$.

Clearly:
\begin{align*}
\left\langle \iota(\alpha),\tilde{\iota}(v)\right\rangle_\times &= \left\langle (\mathrm{pr}_2 \circ \iota)(\alpha), (\widetilde{\mathrm{pr}}_2 \circ \tilde{\iota})(v)\right\rangle - \left\langle (\mathrm{pr}_1 \circ \iota)(\alpha), (\widetilde{\mathrm{pr}}_1 \circ \tilde{\iota})(v)\right\rangle\\
&= \left\langle \mathrm{pr}_1(\alpha), \widetilde{\mathrm{pr}}_1(v)\right\rangle - \left\langle \mathrm{pr}_2(\alpha), \widetilde{\mathrm{pr}}_2(v)\right\rangle\\
& = - \left\langle \alpha , v \right\rangle_{\times}
\end{align*}
\end{proof}

\begin{proposition}
\label{prop:energy_inversion}
Let $\mathbf{L}$ be a Lagrangian satisfying  such that $\mathbf{L} \circ \tilde{\iota} = - \mathbf{L}$ then also $E_\mathbf{L} \circ \tilde{\iota} = - E_\mathbf{L}$.
\end{proposition}

\begin{proof}
Applying the inversion to the definition of the energy given in terms of $\mathbb{F}\mathbf{L}^{\times}$ we have:
\begin{align*}
\left(E_\mathbf{L} \circ \tilde{\iota}\right)(\cdot) &= \left\langle (\mathbb{F}\mathbf{L}^{\times} \circ \tilde{\iota}) (\cdot), \tilde{\iota}(\cdot)\right\rangle_{\times} - (\mathbf{L} \circ \tilde{\iota})(\cdot)\\
&= \left\langle (\iota \circ \mathbb{F}\mathbf{L}^{\times}) (\cdot), \tilde{\iota}(\cdot)\right\rangle_{\times} - (\mathbf{L} \circ \tilde{\iota})(\cdot)
\end{align*}
Applying lemma \ref{lem:inner_prod_inversion} and the inversion property of $\mathbf{L}$ we get:
\begin{align*}
\left(E_\mathbf{L} \circ \tilde{\iota}\right)(\cdot) &= - \left\langle \mathbb{F}\mathbf{L}^{\times}(\cdot), \cdot\right\rangle_{\times} + \mathbf{L}(\cdot)\\
&= - E_\mathbf{L}(\cdot)
\end{align*}
\end{proof}

\begin{corollary}
\label{cor:inversion_Hamiltonian}
Let $\mathbf{L}$ be a regular Lagrangian satisfying the hypothesis of proposition \ref{prop:inversion_commutation}, and define its associated Lagrangian by the expression $\mathbf{H} \circ \mathbb{F}\mathbf{L}^{\times} = E_{\mathbf{L}}$. Then $\mathbf{H} \circ \iota = - \mathbf{H}$.
\end{corollary}

\begin{proof}
Using proposition \ref{prop:inversion_commutation} we have that $\mathbf{H} \circ \iota \circ \mathbb{F}\mathbf{L}^{\times} = \mathbf{H} \circ \mathbb{F}\mathbf{L}^{\times} \circ \tilde{\iota} = E_{\mathbf{L}} \circ \tilde{\iota}$. Applying proposition \ref{prop:energy_inversion} the result follows immediately.
\end{proof}

Finally we can state the following result:
\begin{proposition}\label{proposition-X}
Assume $L$ is a regular Lagrangian, then the Hamiltonian vector field $X_\mathbf{L}$ associated to $\mathbf{L}$ given by
\[
\imath_{X_\mathbf{L}}\omega_\mathbf{L}=\mathrm{d}E_\mathbf{L}
\]
verifies that:
\begin{enumerate}[label=\roman*)]
\item \label{itm:proposition-X-itm1} $X_\mathbf{L}$ is tangent to $\tilde{\epsilon}(TQ)$;
\item \label{itm:proposition-X-itm2} $\left.X_\mathbf{L}\right\vert_{\tilde{\epsilon}(TQ)}=\tilde{\epsilon}_* (X_L)$.
\end{enumerate}
\end{proposition}

\begin{proof}
It is easy to check that $\left(\mathbb{F}\mathbf{L}^{\times}\right)^* \Omega_{Q \times Q} = \omega_\mathbf{L}$. Defining $H = E_{L} \circ \left(\mathbb{F}L\right)^{-1}$ we then get that $\mathbf{H} = E_{\mathbf{L}} \circ \left(\mathbb{F}\mathbf{L}^{\times}\right)^{-1}$. Thus the results of lemma \ref{lemma-1} also apply to $X_\mathbf{L}$ and these results can be brought back to $TQ \times TQ$, proving our claim.
\end{proof}

As with the Hamiltonian formulation we may also include potentials in our description. Again, let $\widetilde{K} : TQ \times TQ \to \mathbb{R}$ be a function such that $\widetilde{K} \circ \tilde{\iota} = -\widetilde{K}$, then $\widetilde{K}$ is a generalized potential and $\mathbf{L}_{\widetilde{K}} = \mathbf{L} - \widetilde{K}$ satisfies $\mathbf{L}_{\widetilde{K}} \circ \tilde{\iota} = -\mathbf{L}_{\widetilde{K}}$.

Given a forced Lagrangian system $(L, F)$ we may define the generalized potential $\widetilde{K}_F: TQ\times TQ\rightarrow {\mathbb R}$ explicitly written as:
\[
\widetilde{K}_F(v_q, V_Q) = \frac{1}{2}\langle F(V_Q), \tau(Q, q)\rangle - \frac{1}{2}\langle F(v_q), \tau(q, Q)\rangle.
\]

Note that if $\mathbf{L}_{\widetilde{K}_F}$ is regular we may obtain a Hamiltonian from its energy as $\widetilde{\mathbf{H}}_{K_F} = E_{\mathbf{L}_{\widetilde{K}_F}} \circ \left(\mathbb{F}\mathbf{L}^{\times}_{\widetilde{K}_F}\right)^{-1}$ but in general it will not be the same Hamiltonian as we defined in the previous section, i.e. 
\begin{equation*}
\widetilde{\mathbf{H}}_{K_F} \neq H \circ \mathrm{pr}_2 - H \circ \mathrm{pr}_1 + \frac{1}{2}\langle F^H \circ \mathrm{pr}_2, \tau \circ \pi_{Q \times Q} \circ \iota \rangle - \frac{1}{2}\langle F^H \circ \mathrm{pr}_1, \tau \circ \pi_{Q \times Q}\rangle
\end{equation*}
in general. This will only be equal if $F$ does not depend on velocities, which means $\mathbb{F}\mathbf{L}^{\times}_{\widetilde{K}_F} = \mathbb{F}\mathbf{L}^{\times}$. This means we cannot directly invoke the result from theorem \ref{theorem-1} to prove that the resulting Euler-Lagrange field coincides with the forced dynamics at the identities and instead we must work a bit more to get the same result. Still at the end of the section we will show that we can actually relate both dynamics obtained from $\mathbf{H}_{K_F}$ and $\widetilde{\mathbf{H}}_{K_F}$.

Let us begin with this regularity result:
\begin{proposition}
If $L$ is regular then $\mathbf{L}_{\widetilde{K}_F}$ is regular in a neighborhood of $\tilde{\epsilon}(TQ)$.
\end{proposition}

\begin{proof}
The proof is essentially the same as that of proposition \ref{proposition-2}.
\end{proof}

Given $\mathbf{L}_{\widetilde{K}} = \mathbf{L} - {\widetilde{K}}$, with $\mathbf{L} = L \circ \widetilde{\mathrm{pr}}_2 - L \circ \widetilde{\mathrm{pr}}_1$, let us reserve $X_{\mathbf{L}}$ for the \emph{free} Euler-Lagrange field, i.e., the vector field that satisfies $$\imath_{X_{\mathbf{L}}} \omega_{\mathbf{L}} = \mathrm{d}E_{\mathbf{L}},$$ and define the vector field $Y_{\widetilde{K}}$ as the one resulting from the decomposition $X_{\mathbf{L}_{\widetilde{K}}} = X_{\mathbf{L}} - Y_{\widetilde{K}}$. Similarly, let us define $\theta_{\widetilde{K}}$ by $\theta_{\mathbf{L}_{\widetilde{K}}} = \theta_{\mathbf{L}} - \theta_{\widetilde{K}}$ and $E_{\widetilde{K}}$ by $E_{\mathbf{L}_{\widetilde{K}}} = E_{\mathbf{L}} - E_{\widetilde{K}}$. Clearly if $\omega_{\widetilde{K}} = - \mathrm{d} \theta_{\widetilde{K}}$ then $\omega_{\mathbf{L}_{\widetilde{K}}} = \omega_{\mathbf{L}} - \omega_{\widetilde{K}}$, and it is not difficult to show that:
\begin{align*}
\imath_{Y_{\widetilde{K}}} \omega_{\mathbf{L}_{\widetilde{K}}} &= \mathrm{d} E_{\widetilde{K}} - \imath_{X_{\mathbf{L}}} \omega_{\widetilde{K}}\\
&= \mathcal{L}_{X_\mathbf{L}} \theta_{\widetilde{K}} - \mathrm{d} \widetilde{K}\; .
\end{align*}
%where in the second line the fact that $\imath_{X_{\mathbf{L}}} \theta_{\widetilde{K}} = \left\langle \mathbb{F}\widetilde{K}^{\times}(v_q, V_Q), (v_q, V_Q)\right\rangle_{\times}$ and Cartan's magic formula have been used.

We can now prove the following theorem:
\begin{theorem}
\label{thm:decomposition_with_K_f}
Given $\mathbf{L}_{\widetilde{K}} = \mathbf{L} - {\widetilde{K}}$ regular in a neighbourhood of $\epsilon(TQ)$, with $\mathbf{L} = L \circ \widetilde{\mathrm{pr}}_2 - L \circ \widetilde{\mathrm{pr}}_1$, whose Hamiltonian vector field $X_{\mathbf{L}_{\widetilde{K}}} = X_{\mathbf{L}} - Y_{\widetilde{K}}$ satisfies:
\begin{equation}
\imath_{X_{\mathbf{L}_{\widetilde{K}}}} \omega_{\mathbf{L}_{\widetilde{K}}} = \mathrm{d} E_{\mathbf{L}_{\widetilde{K}}}
\end{equation}
then:
\begin{enumerate}[label=\roman*)]
\item \label{thm:decomposition_with_K_f-itm0} $X_{\mathbf{L}_{\widetilde{K}}}$ is tangent to $\tilde{\epsilon}(TQ)$;
\item \label{thm:decomposition_with_K_f-itm1} $Y_{\widetilde{K}}$ is vertical and such that $\tilde{\iota}_{*} Y_{\widetilde{K}} = Y_{\widetilde{K}}$.
\end{enumerate}
Furthermore, if $\widetilde{K} = \widetilde{K}_F$, then we have that:
\begin{enumerate}[resume,label=\roman*)]
\item \label{thm:decomposition_with_K_f-itm2} $\left.\imath_{Y_{\widetilde{K}_F}} \omega_{\widetilde{K}_F}\right\vert_{\tilde{\epsilon}(TQ)} = 0$;
\item \label{thm:decomposition_with_K_f-itm3} $\left.\imath_{Y_{\widetilde{K}_F}} \omega_{\mathbf{L}}\right\vert_{\tilde{\epsilon}(TQ)} = \left( \widetilde{\mathrm{pr}}_1^* F-\widetilde{\mathrm{pr}}_2^* F \right)\vert_{\tilde{\epsilon}(TQ)}$.
\item \label{thm:decomposition_with_K_f-itm4} $\left.X_{E_{\mathbf{L}_{\widetilde{K}}}}\right\vert_{\tilde{\epsilon}(TQ)} = \tilde{\epsilon}_* (X_{E_L} + Z^v_F)$
\end{enumerate}
\end{theorem}

\begin{proof}
\begin{enumerate}[label=\roman*)]
\item \label{proof:decomposition_with_K_f-itm0}
Geometrically, given the regularity in a neighborhood of the identities $\omega_{\mathbf{L}_{\widetilde{K}}}$ is non-degenerate there. As $E_{\mathbf{L}_{\widetilde{K}}} \circ \tilde{\iota} = - E_{\mathbf{L}_{\widetilde{K}}}$, we are in a position to apply proposition \ref{prop:tangency_to_identities}, rendering this equivalent to theorem \ref{theorem-1}.\ref{itm:theorem-1-itm1}. Variationally the proof is a consequence of \ref{propo1}.
\item \label{proof:decomposition_with_K_f-itm1} We know that $X_{\mathbf{L}_{\widetilde{K}}}$ and $X_{\mathbf{L}}$ are second order vector fields, as they solve their respective Euler-Lagrange equations. Thus $S(X_{\mathbf{L}_{\widetilde{K}}}) = S(X_{\mathbf{L}}) = \triangle$. Then necessarily $S(Y_{\widetilde{K}}) = 0$, which means it is vertical. As by proposition \ref{proposition-X} both fields satisfy the symmetry property with respect to $\tilde{\iota}$, $Y_{\widetilde{K}}$ must necessarily satisfy it too.
\item \label{proof:decomposition_with_K_f-itm2} For a general $\widetilde{K}$, using Cartan's magic formula we get that $\imath_{Y_{\widetilde{K}}} \omega_{\widetilde{K}} = \mathrm{d}\left( \imath_{Y_{\widetilde{K}}} \theta_{\widetilde{K}}\right) - \mathcal{L}_{Y_{\widetilde{K}}} \theta_{\widetilde{K}}$, but as we have just shown in \ref{proof:decomposition_with_K_f-itm1}, $Y_{\widetilde{K}}$ is vertical, and so $\imath_{Y_{\widetilde{K}}} \theta_{\widetilde{K}}$ vanishes identically. Thus, $\imath_{Y_{\widetilde{K}}} \omega_{\widetilde{K}} = - \mathcal{L}_{Y_{\widetilde{K}}} \theta_{\widetilde{K}}$. Now, if $Y_{\widetilde{K}} = Y_v^i \partial_{v^i} + Y_V^i \partial_{V^i}$ and $\theta_{\widetilde{K}} = \theta^q_{i} \mathrm{d}q^i + \theta^Q_{i} \mathrm{d}Q^i$, then:
\begin{equation*}
\mathcal{L}_{Y_{\widetilde{K}}} \theta_{\widetilde{K}} = \left(Y_v^j \partial_{v^j} + Y_V^j \partial_{V^j}\right) \theta^q_{i} \mathrm{d}q^i + \left(Y_v^j \partial_{v^j} + Y_V^j \partial_{V^j}\right) \theta^Q_{i} \mathrm{d}Q^i
\end{equation*}

In the special case of $\widetilde{K}_F$ we have that $\partial_{V^j} \theta^q_{i} = \partial_{v^j} \theta^Q_{i}= 0$, which reduces the former expression to:
\begin{equation*}
\mathcal{L}_{Y_{\widetilde{K}_F}} \theta_{\widetilde{K}_F} = Y_v^j \partial_{v^j} \theta^q_{i} \mathrm{d}q^i + Y_V^j \partial_{V^j} \theta^Q_{i} \mathrm{d}Q^i
\end{equation*}

Furthermore, $\partial_{v^j} \theta^q_{i} = -\frac{1}{2} \left\langle \frac{\partial^2 F}{\partial v^j \partial v^i}, \tau(q,Q)\right\rangle$, and similarly with $\partial_{V^j} \theta^Q_{i}$. As $\tau \circ \tilde{\epsilon} = 0$, all these terms vanish at the identities.
\item \label{proof:decomposition_with_K_f-itm3} We know that $\imath_{Y_{\widetilde{K}_F}} \omega_{\mathbf{L}} = \mathcal{L}_{X_{\mathbf{L}_{\widetilde{K}_F}}} \theta_{\widetilde{K}_F} - \mathrm{d}\widetilde{K}_F + \imath_{Y_{\widetilde{K}_F}} \omega_{\widetilde{K}_F}$. We have just proven in \ref{proof:decomposition_with_K_f-itm2} that at the identities the last term on the right-hand side vanishes. Thus we only need to worry about the first and second terms. Proceeding as before, we can expand the Lie derivative with $X_{\mathbf{L}_{\widetilde{K}}} = v^i \partial_{q^i} + V^i \partial_{Q^i} + X_v^i \partial_{v^i} + X_V^i \partial_{V^i}$:
\begin{align*}
\mathcal{L}_{X_{\mathbf{L}_{\widetilde{K}}}} \theta_{\widetilde{K}} &= \left(v^j \partial_{q^j} + V^j \partial_{Q^j} + X_v^j \partial_{v^j} + X_V^j \partial_{V^j}\right) \theta^q_{i} \mathrm{d}q^i\\
&+ \left(v^j \partial_{q^j} + V^j \partial_{Q^j} + X_v^j \partial_{v^j} + X_V^j \partial_{V^j}\right) \theta^Q_{i} \mathrm{d}Q^i \\
&+ \theta^q_{j} \partial_{v^i} v^j \mathrm{d}v^i + \theta^Q_{j} \partial_{V^i} V^j \mathrm{d}V^i
\end{align*}
Again, for $\widetilde{K} = \widetilde{K}_F$ we have that $\partial_{V^j} \theta^q_{i} = \partial_{v^j} \theta^Q_{i}= 0$, so simplifying this expression we get:
\begin{align*}
\mathcal{L}_{X_{\mathbf{L}_{\widetilde{K}_F}}} \theta_{\widetilde{K}_F} &= \left(v^j \partial_{q^j} + V^j \partial_{Q^j} + X_v^j \partial_{v^j}\right) \theta^q_{i} \mathrm{d}q^i \\
&+ \left(v^j \partial_{q^j} + V^j \partial_{Q^j} + X_V^j \partial_{V^j}\right) \theta^Q_{i} \mathrm{d}Q^i\\
&+ \theta^q_{i} \mathrm{d}v^i + \theta^Q_{i} \mathrm{d}V^i
\end{align*}
Under the same argument as in \ref{proof:decomposition_with_K_f-itm2}, the terms with derivatives in the $v$ and $V$ variables vanish at the identities. We can also see that $v^j \partial_{q^j} \theta^q_{i} + V^j \partial_{Q^j} \theta^q_{i}$ and its $\theta^Q$ counterpart must also vanish at the identities because:
\begin{equation*}
\partial_{Q}\tau(q,q) = - \partial_{q}\tau(q,q)
\end{equation*}
Thus $\left.\mathcal{L}_{X_{\mathbf{L}_{\widetilde{K}_F}}} \theta_{\widetilde{K}_F}\right\vert_{\tilde{\epsilon}(TQ)} = \left.\left(\frac{\partial \widetilde{K}_F}{\partial v^i} \mathrm{d}v^i + \frac{\partial \widetilde{K}_F}{\partial V^i} \mathrm{d}V^i\right)\right\vert_{\tilde{\epsilon}(TQ)} = 0$. These also coincide with the $v$ and $V$ components of $\mathrm{d}\widetilde{K}$, so we only need to check what happens with the $q$ and $Q$ components.
\begin{align*}
\mathrm{d}\widetilde{K}_F &= \frac{1}{2} \left(\left\langle F(Q,V), \frac{\partial \tau}{\partial q^i}(Q,q)\right\rangle - \left\langle \frac{\partial F}{\partial q^i}(q,v), \tau(q,Q)\right\rangle\right.\\
&- \left.\left\langle F(q,v), \frac{\partial \tau}{\partial q^i}(q,Q)\right\rangle\right) \mathrm{d}q^i\\
&+ \frac{1}{2} \left(\left\langle \frac{\partial F}{\partial Q^i}(Q,V), \tau(Q,q)\right\rangle + \left\langle F(Q,V), \frac{\partial \tau}{\partial Q^i}(Q,q)\right\rangle\right.\\
&- \left.\left\langle F(q,v), \frac{\partial \tau}{\partial Q^i}(q,Q)\right\rangle\right) \mathrm{d}Q^i + ...
\end{align*}

At the identities all terms with a bare $\tau$ vanish, and using the properties of its derivatives, the remaining terms add up together forming $F_i(q,v) \mathrm{d}q^i - F_i(q,v) \mathrm{d}Q^i $.
\item \label{proof:decomposition_with_K_f-itm4} Let us develop the left-hand side of \ref{thm:decomposition_with_K_f-itm3}:
\begin{equation*}
\imath_{Y_{\widetilde{K}_F}} \omega_{\mathbf{L}} = -\frac{\partial^2 L}{\partial v^j \partial v^i} Y_v^j \mathrm{d} q^i + \frac{\partial^2 L}{\partial V^j \partial V^i} Y_V^j \mathrm{d} Q^i
\end{equation*}
Restricting to $\tilde{\epsilon}(TQ)$ and using \ref{proof:decomposition_with_K_f-itm3} we get that:
\begin{equation*}
\left.Y_{\widetilde{K}_F}\right\vert_{\tilde{\epsilon}(TQ)} = -W^{ij} F_j \partial_{v^j} - W^{i j} F_j \partial_{V^j} = \tilde{\epsilon}_*(- Z^v_F)
\end{equation*}
where $W^{i j}$ are the entries of the inverse of the Hessian matrix of $L$, as defined in section \ref{sec:forced_systems}.
\end{enumerate}
\end{proof}

Finding the integral curves of $X_{E_L}+Z^v_F$ is equivalent to solving the forced Euler-Lagrange equations
\[
\frac{\mathrm{d}}{\mathrm{d}t} \left(\frac{\partial L}{\partial \dot{q}^i}\right)- \frac{\partial L}{\partial {q}^i}=F_i(q, \dot{q}^i)\; .
\]
Then from theorem \ref{thm:decomposition_with_K_f} we have that this is also equivalent to solving the unforced Lagrangian system derived by duplication given by $\mathbf{L}_{\widetilde{K}_F}: TQ\times TQ\rightarrow {\mathbb R}$ and restricting the dynamics to $\tilde{\epsilon}(TQ)$; in other words, it is equivalent to solving this system's Euler-Lagrange equations
\begin{eqnarray*}
\frac{\mathrm{d}}{\mathrm{d}t} \left(\frac{\partial \mathbf{L}_{\widetilde{K}_F}}{\partial \dot{q}^i}\right)- \frac{\partial \mathbf{L}_{\widetilde{K}_F}}{\partial {q}^i}&=&0\\
\frac{\mathrm{d}}{\mathrm{d}t} \left(\frac{\partial \mathbf{L}_{\widetilde{K}_F}}{\partial \dot{Q}^i}\right)- \frac{\partial \mathbf{L}_{\widetilde{K}_F}}{\partial {Q}^i}&=&0
\end{eqnarray*}
when restricted to $\tilde{\epsilon}(TQ)=\{(v_q, v_q)\in TQ\times TQ\}$. 

After theorem \ref{thm:decomposition_with_K_f} the following result does not add much more, but gives us a better picture of the difference between the Hamiltonian and Lagrangian side and why the dynamics at the identities coincide:
\begin{theorem}
\label{thm:coincident_dynamics}
Let $(L, F)$ and $(H, F^H)$ be a regular forced Lagrangian system and its associated forced Hamiltonian system, and denote by $\widetilde{\mathbf{H}}_{K_F} = E_{\mathbf{L}_{\widetilde{K}_F}} \circ \left(\mathbb{F}\mathbf{L}^{\times}_{\widetilde{K}_F}\right)^{-1}$ and $\mathbf{H}_{K_F}$ the corresponding generalized Hamiltonians. Then their respective Hamiltonian vector fields $X_{\widetilde{\mathbf{H}}_{K_F}}$ and $X_{\mathbf{H}_{K_F}}$ satisfy that $\left.X_{\widetilde{\mathbf{H}}_{K_F}}\right\vert_{\epsilon(T^*Q)} = \left.X_{\mathbf{H}_{K_F}}\right\vert_{\epsilon(T^*Q)}$.
\end{theorem}

\begin{proof}
Working on the Lagrangian side we know that $\mathbf{H}_{K_F} \circ \mathbb{F}\mathbf{L}^{\times} = E_{\mathbf{L}} + \widetilde{K}_F$, where $\mathbb{F}\mathbf{L}^{\times}$ is the Legendre transformed induced by the free Lagrangian. This means that we have the following two concurrent dynamics:
\begin{align}
\imath_{X_{\mathbf{L}_{\widetilde{K}_F}}} \omega_{\mathbf{L}_{\widetilde{K}_F}} &= \mathrm{d} E_{\mathbf{L}_{\widetilde{K}_F}}\\
\imath_{\widehat{X}_{\mathbf{L}_{\widetilde{K}_F}}} \omega_{\mathbf{L}} &= \mathrm{d} \left( E_{\mathbf{L}} + \widetilde{K}_F\right)
\end{align}

The respective transformed versions of the vector fields $X_{\mathbf{L}_{\widetilde{K}_F}}$ and $\widehat{X}_{\mathbf{L}_{\widetilde{K}_F}}$ are $X_{\widetilde{\mathbf{H}}_{K_F}}$ and $X_{\mathbf{H}_{K_F}}$. Clearly both $X_{\mathbf{L}_{\widetilde{K}_F}}$ and $\widehat{X}_{\mathbf{L}_{\widetilde{K}_F}}$ can be decomposed into $X_{\mathbf{L}} - Y_{\widetilde{K}_F}$ and $X_{\mathbf{L}} - \widehat{Y}_{\widetilde{K}_F}$ respectively. We are left with:
\begin{align}
\imath_{Y_{\widetilde{K}_F}} \omega_{\mathbf{L}_{\widetilde{K}_F}} &= \mathcal{L}_{X_\mathbf{L}} \theta_{\widetilde{K}_F} - \mathrm{d} \widetilde{K}_F\\
\imath_{\widehat{Y}_{\widetilde{K}_F}} \omega_{\mathbf{L}} &= - \mathrm{d} \widetilde{K}_F
\end{align}

As we saw in proposition \ref{thm:decomposition_with_K_f}.\ref{thm:decomposition_with_K_f-itm2}, $\left.\imath_{Y_{\widetilde{K}_F}} \omega_{\widetilde{K}_F}\right\vert_{\tilde{\epsilon}(TQ)} = 0$, so we are left with:
\begin{align}
\imath_{Y_{\widetilde{K}_F}} \omega_{\mathbf{L}} &= \mathcal{L}_{X_\mathbf{L}} \theta_{\widetilde{K}_F} - \mathrm{d} \widetilde{K}_F
\end{align}
and we saw in \ref{thm:decomposition_with_K_f}.\ref{thm:decomposition_with_K_f-itm3} $\left.\mathcal{L}_{X_\mathbf{L}} \theta_{\widetilde{K}_F}\right\vert_{\tilde{\epsilon}(TQ)} = 0$, thus $\left.Y_{\widetilde{K}_F}\right\vert_{\tilde{\epsilon}(TQ)} = \left.\widehat{Y}_{\widetilde{K}_F}\right\vert_{\tilde{\epsilon}(TQ)}$. As both $\mathbb{F}\mathbf{L}^{\times}_{\widetilde{K}_F}$ and $\mathbb{F}\mathbf{L}^{\times}$ coincide at the identities, then so will $X_{\widetilde{\mathbf{H}}_{K_F}}$ and $X_{\mathbf{H}_{K_F}}$, proving our claim.
\end{proof}

The two terms that must vanish at the identities for both dynamics to coincide, $\imath_{Y_{\widetilde{K}_F}} \omega_{\widetilde{K}_F}$ and $\mathcal{L}_{X_\mathbf{L}} \theta_{\widetilde{K}_F}$, amount to the condition that $\mathcal{L}_{X_{\mathbf{L}_{\widetilde{K}_F}}} \theta_{\widetilde{K}_F}$ vanishes at the identities. This is still true for any $\widetilde{K}$ such that $K \circ \mathbb{F}\mathbf{L}^{\times} = \widetilde{K}$.

\section{Variational order for forced discrete Lagrangian systems}
\subsection{Introduction to discrete mechanics}

We will consider $Q\times Q$ as a discrete version of $TQ$ and therefore $Q\times Q\times Q \times Q$ as a discrete analogue of $TQ \times TQ$, see \cite{marsden-west}. Instead of curves on $Q$, the solutions are replaced by sequences of points on $Q$. If we fix some $N\in \mathbb{N}$ then we use the notation 
\begin{equation*}
\mathcal{C}_d(Q)=\left\{ q_d:\left\{ k \right\}_{k=0}^N \longrightarrow Q \right\}
\end{equation*}
for the set of possible solutions, which can be identified with the manifold $Q\times \stackrel{(N+1)}{\cdots} \times Q$.
Define a functional, the discrete action map, on the space of sequences $\mathcal{C}_d(Q)$ by
\begin{equation*}
S_d (q_d)=\sum_{k=0}^{N-1}L_d(q_k,q_{k+1}), \quad q_d\in \mathcal{C}_d(Q) .
\end{equation*}
If we consider variations of $q_d$ with fixed end points $q_0$ and $q_N$ and extremize $S_d$ over $q_1,\ldots,q_{N-1}$, we obtain the discrete Euler-Lagrange equations (DEL)
\begin{equation}
\label{Eq:DEL}
D_1 L_d (q_k,q_{k+1})+D_2L_d(q_{k-1},q_k)=0 \quad \mbox{for all} \quad k=1,\ldots,N-1 \, ,
\end{equation}
where $D_1L_d(q_{k-1},q_k)\in T^*_{q_{k-1}}Q$ and $D_2L_d(q_{k-1},q_k)\in T^*_{q_k}Q$ correspond to 
$\mathrm{d}L_d(q_{k-1},q_k)$ under the identification $T^*_{(q_{k-1},q_k)}(Q\times Q)\cong T^*_{q_{k-1}}Q\times T^*_{q_k}Q$.

If $L_d$ is regular, that is, $D_{12}L_d$ is regular, then we obtain a well defined discrete Lagrangian map
\begin{equation*}
\begin{array}{cccc}
F_{L_d}: & Q\times Q & \longrightarrow & Q \times Q \\
& (q_{k-1},q_k) & \longmapsto & (q_k,q_{k+1}(q_{k-1},q_k)) \, ,
\end{array}
\end{equation*}
where $q_{k+1}$ is the unique solution of~\eqref{Eq:DEL} for the given pair $(q_{k-1},q_k)$. We can further assure that the discrete Lagrangian map is invertible so that it is possible to write $q_{k-1}=q_{k-1}(q_k,q_{k+1})$, see \cite[Theorem 1.5.1]{marsden-west}.

In this setting we can define two discrete Legendre transformations
\begin{equation*}
\mathbb{F}^{+}L_{d},\mathbb{F}^{-}L_{d}:Q\times Q\longrightarrow T^*Q  ,
\end{equation*}
since each projection is equally eligible for the base point. They can be defined as
\begin{eqnarray*}
&& \mathbb{F}^{+}L_{d}(q_{k-1},q_{k})= (q_{k},D_2L_d(q_{k-1},q_{k})) \, ,\\
&& \mathbb{F}^{-}L_{d}(q_{k-1},q_{k}) = (q_{k-1},-D_1L_d(q_{k-1},q_{k})) \, .
\end{eqnarray*}
It holds that $(\mathbb{F}^{+}L_{d})^*\omega_Q=(\mathbb{F}^{-}L_{d})^*\omega_Q=:\Omega_{L_d}$, with local expression
\begin{equation*}
\Omega_{L_d}(q_{k-1},q_k)=\frac{\partial^2 L_d}{\partial q_{k-1}^i \partial q_k^j}\mathrm{d}q_{k-1}^i \wedge \mathrm{d}q_k^j .
\end{equation*}

We can also define the evolution of the discrete system on the Hamiltonian side, $\widetilde{F}_{L_d}:T^*Q \longrightarrow T^*Q$, by any of the formulas
$$
\widetilde{F}_{L_d}=\mathbb{F}^{+}L_{d}\circ (\mathbb{F}^{-}L_{d})^{-1}=\mathbb{F}^{+}L_{d}\circ F_{L_d} \circ (\mathbb{F}^{+}L_{d})^{-1}=\mathbb{F}^{-}L_{d}\circ F_{L_d} \circ (\mathbb{F}^{-}L_{d})^{-1} \, ,
$$
because of the commutativity of the following diagram:
$$
\xymatrix{ 
(q_{k-1},q_k) \ar[rr]^{F_{L_d}} \ar[dr]_{\mathbb{F}^{+}L_{d}} & & (q_{k},q_{k+1}) \ar[ld]_{\mathbb{F}^{-}L_{d}} \ar[rd]_{\mathbb{F}^{+}L_{d}} \ar[rr]^{F_{L_d}} & &  (q_{k+1},q_{k+2}) \ar[ld]_{\mathbb{F}^{-}L_{d}} \\
& (q_k,p_k) \ar[rr]_{\widetilde{F}_{L_d}} & & (q_{k+1},p_{k+1}) &
}
$$
The discrete Hamiltonian map $\widetilde{F}_{L_d}:(T^*Q,\omega_Q) \longrightarrow (T^*Q,\omega_Q)$ is symplectic. Therefore the submanifold 
\begin{eqnarray*}
\left( q_k, p_k, q_{k+1}, p_{k+1} \right) &=& \left( q_k, \mathbb{F}^{-}L_{d}(q_k,q_{k+1}), q_{k+1}, \mathbb{F}^{-}L_{d}(q_{k+1},q_{k+2}) \right) \\
&=& \left( q_k, \mathbb{F}^{+}L_{d}(q_{k-1},q_{k}), q_{k+1}, \mathbb{F}^{+}L_{d}(q_{k},q_{k+1}) \right)
\end{eqnarray*}
is Lagrangian in $(T^*Q\times T^*Q,\Omega_Q)$, where $\Omega_{Q}:=\beta_{T^*Q}^{*}\omega_{Q}-\alpha_{T^*Q}^{*}\omega_{Q}$ is a symplectic form and $\alpha_{T^*Q},\beta_{T^*Q}:T^*Q\times T^*Q \longrightarrow T^*Q$ denote the projections onto the first and second factor respectively.

So far we have taken as the starting point a discrete Lagrangian $L_d : Q\times Q \longrightarrow \mathbb{R}$. However, if we start with a continuous Lagrangian and take an appropriate discrete Lagrangian then the DEL equations become a geometric integrator for the continuous Euler-Lagrange system, known as a variational integrator. 
Hence, given a regular Lagrangian function $L:TQ \longrightarrow \mathbb{R}$, we define a discrete Lagrangian $L_d: Q\times Q \times \mathbb{R} \longrightarrow \mathbb{R}$ as an approximation to the action of the continuous Lagrangian. More precisely, for a regular Lagrangian $L$, and appropriate $h,q_0,q_1$, we can define the exact discrete Lagrangian as
\begin{equation*}
L_d^e(q_0,q_1,h) = \int_0^h L(q_{0,1}(t),\dot{q}_{0,1}(t))\mathrm{d}t,
\end{equation*}
where $q_{0,1}(t)$ is the unique solution of the Euler-Lagrange equations for $L$ satisfying $q_{0,1}(0)=q_0$ and $q_{0,1}(h)=q_1$, see \cite{hartman,MMM3}. 
Then for a sufficiently small $h$, the solutions of the DEL equations for $L_d^e$ lie on the solutions of the Euler-Lagrange equations for $L$, see \cite[Theorem 1.6.4]{marsden-west}.

In practice, $L_d^e(q_0,q_1,h)$ will not be explicitly given. Therefore we will take 
$$
L_d(q_0,q_1,h) \approx L_d^e(q_0,q_1,h)\, ,
$$
using some quadrature rule. We obtain symplectic integrators in this way, see \cite{PatrickCuell}.

Now we recall the result of \cite{PatrickCuell} and 
\cite{marsden-west} for a discrete Lagrangian
$L_d\colon Q\times Q\rightarrow {\mathbb R}$.

\begin{definition}
Let $L_{d}\colon Q\times Q\rightarrow {\mathbb R}$ be a discrete Lagrangian. We say that
$L_{d}$ is a discretization of order $r$ if there exist an
open subset $U_{1}\subset TQ$ with compact closure and
constants $C_1>0$, $h_1>0$ so that

\begin{equation*}%\label{orderlagrangian}
\lVert L_{d}(q(0),q(h))-L_{d}^{e}(q(0),q(h))\rVert\leq C_{1}h^{r+1}
\end{equation*} for all solutions $q(t)$ of the second-order Euler--Lagrange equations with initial conditions $(q_0,\dot{q}_0)\in U_1$ and for all $h\leq h_1$.
\end{definition}

Following \cite{marsden-west, PatrickCuell}, we have the next result about the order of a variational integrator.

\begin{theorem}
If $\widetilde{F}_{L_d}$ is the evolution map of an order $r$
discretization $L_d\colon Q\times Q\rightarrow {\mathbb R}$ of the exact discrete
Lagrangian $L_d^{e}\colon Q\times Q\rightarrow {\mathbb R}$, then
\[\widetilde{F}_{L_d}=\widetilde{F}_{L_{d}^{e}}+\mathcal{O}(h^{r+1}).\]
In other words, $\widetilde{F}_{L_d}$ gives an integrator of order
$r$ for $\widetilde{F}_{L_{d}^{e}}=F_{H}^{h}$.
\end{theorem}

Note that given a discrete Lagrangian $L_{d}\colon Q\times Q\to\R$ its
order can be calculated by expanding the expressions for
$L_d(q(0),q(h))$ in a Taylor series in $h$
and comparing this to the same expansions for the exact Lagrangian.
If the series agree up to $r$ terms, then the discrete Lagrangian is
of order $r$.

\section{Discrete Lagrangian dynamics obtained by duplication}
We have a regular system defined by $\mathbf{L}_{\widetilde{K}}: TQ\times TQ\rightarrow {\mathbb R}$, now we consider a discretization of this Lagrangian 
\[
\mathbf{L}^d_{\widetilde{K}}: Q\times Q\times Q\times Q\rightarrow {\mathbb R}
\]
such that $\mathbf{L}^d_{\widetilde{K}}=-\mathbf{L}^d_{\widetilde{K}}\circ \tilde{\iota}_d$ where $\tilde{\iota}_d:  Q\times Q\times Q\times Q \rightarrow Q\times Q\times Q\times Q$ is the inversion defined by
\[
\tilde{\iota}_d(q_k, q_{k+1}, Q_k, Q_{k+1})=(Q_k, Q_{k+1}, q_k, q_{k+1})\; .
\]

Additionally define the identity map $\tilde{\epsilon}_d: Q\times Q\rightarrow Q\times Q\times Q\times Q$ by
\[
\tilde{\epsilon}_d (q_k, q_{k+1})=(q_k, q_{k+1}, q_{k}, q_{k+1})\; .
\]

\begin{theorem}
The flow $F_{\mathbf{L}^d_{\widetilde{K}}}:Q\times Q\times Q\times Q\rightarrow Q\times Q\times Q\times Q$ defined by a discrete Lagrangian $\mathbf{L}^d_{\widetilde{K}}: Q\times Q\times Q\times Q\rightarrow {\mathbb R}$ verifying that $\mathbf{L}^d_{\widetilde{K}}=-\mathbf{L}^d_{\widetilde{K}}\circ \tilde{\iota}_d$ restricts to $\tilde{\epsilon}_d(Q\times Q)$, that is, 
\[
F_{\mathbf{L}^d_{\widetilde{K}}}\circ \tilde{\epsilon}_d(Q\times Q)\in \tilde{\epsilon}_d(Q\times Q)\; .
\]
\end{theorem}
\begin{proof}
The proof is a consequence of Proposition \ref{propo2}.
\end{proof}

\section[Variational order for forced Lagrangian systems]{Main result. Variational order for forced Lagrangian systems}

Now, we are in a position to state the main result of this paper. 
\begin{theorem}\label{main-theorem}
Let $(L,F)$ be a forced Lagrangian system. Derive from it the extended regular Lagrangian $\mathbf{L}_{\widetilde{K}_F}: TQ\times TQ\rightarrow {\mathbb R}$ and consider an order $r$ discretization $\mathbf{L}^d_{\widetilde{K}_F}: TQ\times TQ\rightarrow {\mathbb R}$ of the exact Lagrangian 
\[
\mathbf{L}^e_{\widetilde{K}_F} (q_0, Q_0, q_1, Q_1)=\int_0^h\mathbf{L}_{\widetilde{K}_F} (q_{0,1}(t), \dot{q}_{0,1}(t), Q_{0,1}(t), \dot{Q}_{0,1}(t))\; dt
\]
where $t\rightarrow (q_{0,1}(t), \dot{q}_{0,1}(t), Q_{0,1}(t))$ the unique solution of the Euler-Lagrange equations for $\mathbf{L}_{\widetilde{K}_F}$ satisfying 
${q}_{0,1}(0)=q_0, Q_{0,1}(0)=Q_0, {q}_{0,1}(h)=q_1, Q_{0,1}(h)=Q_1$ and satisfying additionally that $\mathbf{L}_{\widetilde{K}_F}^d\circ \tilde{i}_d=-\mathbf{L}_{\widetilde{K}_F}^d$. 
Then, the discrete Euler-Lagrange equations of $\mathbf{L}^d_{\widetilde{K}_F}$ restricted to $\tilde{\epsilon}_d (Q\times Q)$ give us a numerical integrator of order $r$ for the flow of the forced Lagrangian system $(L,F)$. 
\end{theorem}

\begin{example}
{\rm 
As an example consider a Lagrangian $L: {\mathbb R}^{2n}\rightarrow {\mathbb R}$: 
\[
L(q, v)=\frac{1}{2} v^T M v-\frac{1}{2}q^T K q\\; ,
\]

and a dissipation force $F(q^i, v^i)=(q^i, -D_{i j} v^j)$

The forced Euler-Lagrange equations are: 
\[
M\ddot{q}(t)+D\dot{q}(t)+Kq(t)=0\; .
\]

We will derive the extended Lagrangian $\mathbf{L}_{\widetilde{K}_{F}}: {\mathbb R}^{4n}\rightarrow {\mathbb R}$. For that, we consider the function: 
\[
\widetilde{K}_{F}(q, v, Q, V)= - \frac{1}{2} D V\cdot (q-Q) + \frac{1}{2} D v\cdot (Q-q) = \left[\frac{D}{2}\left(V+v\right)\right]^T(Q-q)
\]
Then
\begin{align*}
\mathbf{L}_{\widetilde{K}_{F}} (q, v, Q, V) &= L(Q, V) - L(q, v)  - \widetilde{K}_{F}(q, v, Q, V)\\
&=\frac{1}{2} V^T M V-\frac{1}{2}Q^T K Q-\frac{1}{2} v^T M v+\frac{1}{2}q^T K q - \left[\frac{D}{2}\left(V+v\right)\right]^T(Q-q)
\end{align*}
By construction $\mathbf{L}_{\widetilde{K}_{F}}(q, v, Q, V) = -\mathbf{L}_{\widetilde{K}_{F}}(Q, V, q, v)$. 

Let us discretize by using the so-called ``midpoint rule'"
\[
q\approx \frac{q_0+q_1}{2}\, ,\ \dot{q}\approx \frac{q_1-q_0}{h}
\]
which leads to:
\begin{eqnarray*}
\mathbf{L}_{\widetilde{K}_{F}}^d &=&
\frac{h}{2} \left(\frac{Q_{k+1}-Q_k}{h}\right)^T M \left(\frac{Q_{k+1}-Q_k}{h}\right)
-\frac{h}{2}\left(\frac{Q_{k}+Q_{k+1}}{2}\right)^T K \left(\frac{Q_{k}+Q_{k+1}}{2}\right)\\
&&-\frac{h}{2} \left(\frac{q_{k+1}-q_k}{h}\right)^T M \left(\frac{q_{k+1}-q_k}{h}\right)
+\frac{h}{2}\left(\frac{q_{k}+q_{k+1}}{2}\right)^T K \left(\frac{q_{k}+q_{k+1}}{2}\right)\\
&& -\frac{h}{2}\left[D \left(\frac{Q_{k+1}-Q_k}{h}+\frac{q_{k+1} - q_k}{h}\right)\right]^T \left(\frac{Q_{k}+Q_{k+1}}{2}-\frac{q_{k}+q_{k+1}}{2}\right)
\end{eqnarray*}
Observe that 
\[
\mathbf{L}_{\widetilde{K}_{F}}^d(q_k, Q_k, q_{k+1}, Q_{k+1})=-\mathbf{L}_{\widetilde{K}_{F}}^d(Q_k, q_k, Q_{k+1}, q_{k+1})
\]
Therefore, from theorem \ref{main-theorem} this leads to a second order method restricting the discrete Euler-Lagrange equations to 
$\tilde{\epsilon}_d(Q\times Q)$. The resulting equations are not very surprising:
\[
M \left(\frac{Q_{k+2}-2Q_{k+1}+Q_k}{h^2}\right) + D\left(\frac{Q_{k+2}-Q_k}{2h}\right) + K \left(\frac{Q_{k+2}+2Q_{k+1}+Q_k}{4}\right) =0\; .
\]
}
\end{example}

The following results provide a purely variational base for the exact discrete forcing offered by Marsden and West, and show that usual Runge-Kutta type discretization schemes provide the same results as in their article.

\begin{proposition}
\label{prop:exact_discrete_lagrangian_contribution}
The exact discrete Lagrangian defined by $\mathbf{L}^e_{\widetilde{K}_F}(u_q, v_{q'})$ at the identities is equivalent to two copies of the one defined in \cite[eq.(3.2.7)]{marsden-west}.
\end{proposition}

\begin{proof}
The corresponding parts for $L$ need not be checked as they correspond trivially to those of eq.(3.2.7a) with the adequate change of notation. It remains to show that ${\widetilde{K}_F}$ generates the \emph{exact discrete forces} $f_d^{e+}$, $f_d^{e-}$.

The contribution of some $K$ to the exact discrete Lagrangian is:
\begin{equation}
\widetilde{K}_d^e(q_0,Q_0,q_1,Q_1) = \int_{0}^h \widetilde{K}(q(t),v(t),Q(t),V(t))\mathrm{d}t
\end{equation}
where $t\rightarrow (q(t), Q(t))\in Q\times Q$ is the unique solution for ${\mathbf L}_{\widetilde{K}_F}$ with boundary conditions $q(0)=q_0, Q(0)=Q_0, q(h)=q_1, Q(h)=Q_1$.

In the case where $\widetilde{K} = \widetilde{K}_F$, differentiating $\widetilde{K}_d$ with respect to $q_0$ we get:
\begin{align}
&D_1 {\widetilde{K}_{F,d}}^e(q_0,Q_0,q_1,Q_1)\\
&= \int_{0}^h \left[D_1 {\widetilde{K}_F} \cdot \frac{\partial q(t)}{\partial q_0} + D_2 {\widetilde{K}_F} \cdot \frac{\partial v(t)}{\partial q_0} + D_3 {\widetilde{K}_F} \cdot \frac{\partial Q(t)}{\partial q_0} + D_4 {\widetilde{K}_F} \cdot \frac{\partial V(t)}{\partial q_0}\right]\mathrm{d}t\nonumber
\end{align}
where:
\begin{align*}
D_1 {\widetilde{K}_F} &= \frac{1}{2} \left[\left\langle F(Q,V), D_2 \tau(Q,q)\right\rangle - \left\langle D_1 F(q,v), \tau(q,Q)\right\rangle - \left\langle F(q,v), D_1 \tau(q,Q)\right\rangle \right]\\
D_2 {\widetilde{K}_F} &= - \frac{1}{2} \left\langle D_2 F(q,v), \tau(q,Q)\right\rangle\\
D_3 {\widetilde{K}_F} &= \frac{1}{2} \left[\left\langle D_1 F(Q,V), \tau(Q,q)\right\rangle + \left\langle F(Q,V), D_1 \tau(Q,q)\right\rangle - \left\langle F(q,v), D_2 \tau(q,Q)\right\rangle\right]\\
D_4 {\widetilde{K}_F} &= \frac{1}{2} \left\langle D_2 F(Q,V), \tau(Q,q)\right\rangle
\end{align*}

Similar expressions are found after differentiation with respect to $q_0'$, $q_1$ and $q_1'$. Now, when restricted to the identities, we find that:
\begin{align*}
\tilde{\epsilon}^* D_1 {\widetilde{K}_F} &= \left\langle F(q, v), \mathrm{id}_{T_q Q}\right\rangle\\
\tilde{\epsilon}^* D_2 {\widetilde{K}_F} &= 0\\
\tilde{\epsilon}^* D_3 {\widetilde{K}_F} &= -\left\langle F(q, v), \mathrm{id}_{T_q Q}\right\rangle\\
\tilde{\epsilon}^* D_4 {\widetilde{K}_F} &= 0
\end{align*}
where we used the fact that $\tau(q,q) = 0_q$ and $\tilde{\epsilon}^* D_2 \tau = - \tilde{\epsilon}^* D_1 \tau = - \mathrm{id}_{T_q Q}$. This leads to:
\begin{align}
\tilde{\epsilon}_d^{\,*} D_1 {\widetilde{K}_{F,d}}^e &= - \tilde{\epsilon}_d^{\,*} D_2 {\widetilde{K}_{F,d}}^e = \int_0^h F(q(t),u(t)) \cdot \frac{\partial q(t)}{\partial q_0} \mathrm{d}t = f_d^{e-}\\
\tilde{\epsilon}_d^{\,*} D_3 {\widetilde{K}_{F,d}}^e &= - \tilde{\epsilon}_d^{\,*} D_4 {\widetilde{K}_{F,d}}^e = \int_0^h F(q(t),u(t)) \cdot \frac{\partial q(t)}{\partial q_1} \mathrm{d}t = f_d^{e+}
\end{align}

Putting everything together we find two copies of the forced discrete equations with opposite sign, which is what we set to prove.
\end{proof}

%%%%%%%%%%%%%%%%%%%%%%%%%%%%%%%%%%%%%%%%%%%%%%%%%%%%%%%%%%%%%%%%%%%%%%%%%%%%%%%%%%%%%%%%%%%%%%%%%%

\begin{proposition}
\label{prop:discretisation_contribution}
Let $\gamma^i : \mathbb{R} \times Q \times Q \to TQ$, $i = 1, ..., s$, be differentiable discretisation functions, and let us use for convenience the notation $\tau_Q \circ \gamma^i(t,q_0, q_1) = (f^i(t,q_0,q_1), g^i(t,q_0,q_1))$. Let also $(b_i, c_i)$ be some quadrature coefficients such that the conservative discrete Lagrangian is approximated as:
\begin{equation}
L_d(q_0,q_1) = h \sum_{i=1}^s b_i L \circ \gamma^i(c_i,q_0,q_1).
\end{equation}
Then the contribution of $\widetilde{K}_F$, as defined in Proposition \ref{prop:exact_discrete_lagrangian_contribution}, to the discrete Lagrangian $\mathbf{L}_{\widetilde{K}_F}$ at the identities becomes:
\begin{align}
\tilde{\epsilon}_d^{\,*} \frac{\partial {\widetilde{K}_F^d}}{\partial q_0} &= -h \sum_{i=1}^s b_i \left\langle F \circ \gamma^i(c_i,q_0,q_1), \frac{\partial f^i(c_i,q_0,q_1)}{\partial q_0}\right\rangle\\
\tilde{\epsilon}_d^{\,*} \frac{\partial {\widetilde{K}_F^d}}{\partial Q_0} &= h \sum_{i=1}^s b_i \left\langle F \circ \gamma^i(c_i,q_0,q_1), \frac{\partial f^i(c_i,q_0,q_1)}{\partial q_0}\right\rangle\\
\tilde{\epsilon}_d^{\,*} \frac{\partial {\widetilde{K}_F^d}}{\partial q_1} &= -h \sum_{i=1}^s b_i \left\langle F \circ \gamma^i(c_i,q_0,q_1), \frac{\partial f^i(c_i,q_0,q_1)}{\partial q_1}\right\rangle\\
\tilde{\epsilon}_d^{\,*} \frac{\partial {\widetilde{K}_F^d}}{\partial Q_1} &= h \sum_{i=1}^s b_i \left\langle F \circ \gamma^i(c_i,q_0,q_1), \frac{\partial f^i(c_i,q_0,q_1)}{\partial q_1}\right\rangle
\end{align}
\end{proposition}

\begin{proof}
For the contribution of ${\widetilde{K}_{F}}$ to the discrete Lagrangian we have:
\begin{equation}
{\widetilde{K}_F^d}(q_0,Q_0,q_1,Q_1) = h \sum_{i=1}^s b_i \widetilde{K}(\gamma^i(c_i,q_0,q_1),\gamma^i(c_i,Q_0,Q_1))
\end{equation}

Differentiating with respect to $q_0$, $q_0'$, $q_1$ and $q_1'$ we have:
\begin{align}
\frac{\partial {\widetilde{K}_F^d}}{\partial q_0} &= h \sum_{i=1}^s b_i \left[ D_1 K \cdot \frac{\partial f^i(c_i,q_0,q_1)}{\partial q_0} + D_2 K \cdot \frac{\partial g^i(c_i,q_0,q_1)}{\partial q_0}\right]\\
\frac{\partial {\widetilde{K}_F^d}}{\partial Q_0} &= h \sum_{i=1}^s b_i \left[ D_3 K \cdot \frac{\partial f^i(c_i,Q_0,Q_1)}{\partial Q_0} + D_4 K \cdot \frac{\partial g^i(c_i,Q_0,Q_1)}{\partial Q_0}\right]\\
\frac{\partial {\widetilde{K}_F^d}}{\partial q_1} &= h \sum_{i=1}^s b_i \left[ D_1 K \cdot \frac{\partial f^i(c_i,q_0,q_1)}{\partial q_1} + D_2 K \cdot \frac{\partial g^i(c_i,q_0,q_1)}{\partial q_1}\right]\\
\frac{\partial {\widetilde{K}_F^d}}{\partial Q_1} &= h \sum_{i=1}^s b_i \left[ D_3 K \cdot \frac{\partial f^i(c_i,Q_0,Q_1)}{\partial Q_1} + D_4 K \cdot \frac{\partial g^i(c_i,Q_0,Q_1)}{\partial Q_1}\right]
\end{align}
where $D_i K$ are the same as those of proposition \ref{prop:exact_discrete_lagrangian_contribution} with $q^i(t) = f^i(t,q_0,q_1)$, $v^i(t) = g^i(t,q_0,q_1)$,  $\left(Q\right)^i(t) = f^i(t,Q_0,Q_1)$, $V^i(t) = g^i(t,Q_0,Q_1)$. Restriction to the identities proves our claim.
\end{proof}

\begin{example} Let us choose our discretisation to be:
\begin{equation}
L_d^{\alpha}(q_0,q_1) = h L\left((1-\alpha)q_0 + \alpha q_1, \frac{q_1 - q_0}{h}\right)
\end{equation}
as in \cite[example 3.2.2]{marsden-west}. This results in:
\begin{align}
\tilde{\epsilon}_d^{\,*} \frac{\partial {\widetilde{K}_{F,d}}^{\alpha}}{\partial q_0} &= h (1 - \alpha) F\left((1-\alpha)q_0 + \alpha q_1, \frac{q_1 - q_0}{h}\right)\\
\tilde{\epsilon}_d^{\,*} \frac{\partial {\widetilde{K}_{F,d}}^{\alpha}}{\partial Q_0} &= - h (1 - \alpha) F\left((1-\alpha)q_0 + \alpha q_1, \frac{q_1 - q_0}{h}\right)\\
\tilde{\epsilon}_d^{\,*} \frac{\partial {\widetilde{K}_{F,d}}^{\alpha}}{\partial q_1} &= h \alpha F\left((1-\alpha)q_0 + \alpha q_1, \frac{q_1 - q_0}{h}\right)\\
\tilde{\epsilon}_d^{\,*} \frac{\partial {\widetilde{K}_{F,d}}^{\alpha}}{\partial Q_1} &= - h \alpha F\left((1-\alpha)q_0 + \alpha q_1, \frac{q_1 - q_0}{h}\right)
\end{align}
which coincides with their result.
\end{example}

\subsection{Numerical tests}
For our numerical tests we have chosen a well-known system composed of two coupled van der Pol oscillators (cf. \cite[eq.(6.38)]{Scheck}). Remember that a single dimensionless van der Pol oscillator is described by the differential equation:
\begin{equation*}
\ddot{q} - \left(\epsilon - q^2\right) \dot{q} + q = 0
\end{equation*}
where $\epsilon$ is a parameter related to the damping of the system.

The dimensionless system we are going to study can be thought to be composed of two coupled harmonic oscillators with slightly differing natural frequencies under the action of non-linear forcing. Its configuration manifold is $\mathbb{T} \times \mathbb{T} = \mathbb{T}^2$, with velocity phase space $T\mathbb{T}^2$ where we will use local coordinates $(q_1, q_2, v_1, v_2)$, and the Lagrangian describing the non-forced part $L: T\mathbb{T}^2 \to \mathbb{R}$ is:
\begin{equation*}
L = \frac{1}{2}\left(v_1^2 + v_2^2\right) - \frac{1}{2}\left[q_1^2 + \left(1 + \rho\right)q_2^2\right] - \lambda \left(q_1 - q_2\right)^2
\end{equation*}
where $\rho$ accounts for the deviation of $q_2$ from the natural frequency of $q_1$, and $\lambda$ measures the intensity of the coupling between both oscillators. The van der Pol force acting on this system is $F = \left(\epsilon - q_1^2\right) v_1 \mathrm{d} q_1 + \left(\epsilon - q_2^2\right) v_2 \mathrm{d} q_2$. As our configuration space is flat, $\tau(q,Q) = Q - q$, and the generalized potential $K$ is:
\begin{equation*}
\widetilde{K}_F = \frac{1}{2} \sum_{i = 1}^2 \left[(\epsilon - q_i^2) v_i + (\epsilon - Q_i^2) V_i\right] \left(q_i-Q_i\right)
\end{equation*}
Note that for such an $L$ and $\widetilde{K}_F$, at the identities we have that $v_i = p_i$, $i = 1,2$, so they are interchangeable.

\begin{figure}[h!]
  \centering
  \begin{subfigure}[b]{0.45\textwidth}
  	\includegraphics[scale=0.45]{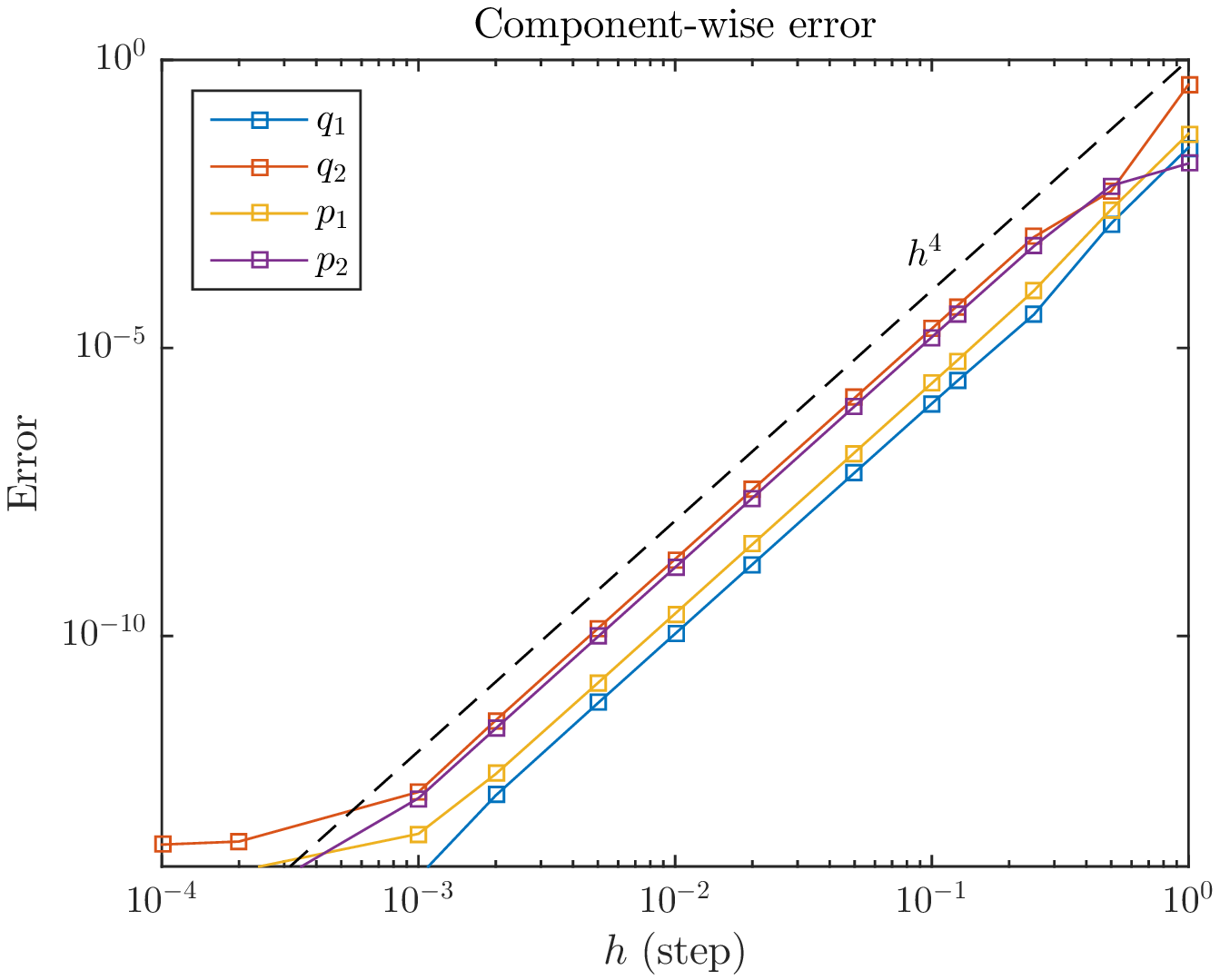}
  \end{subfigure}
  \hspace{0.5cm}
  \begin{subfigure}[b]{0.45\textwidth}
  	\includegraphics[scale=0.45]{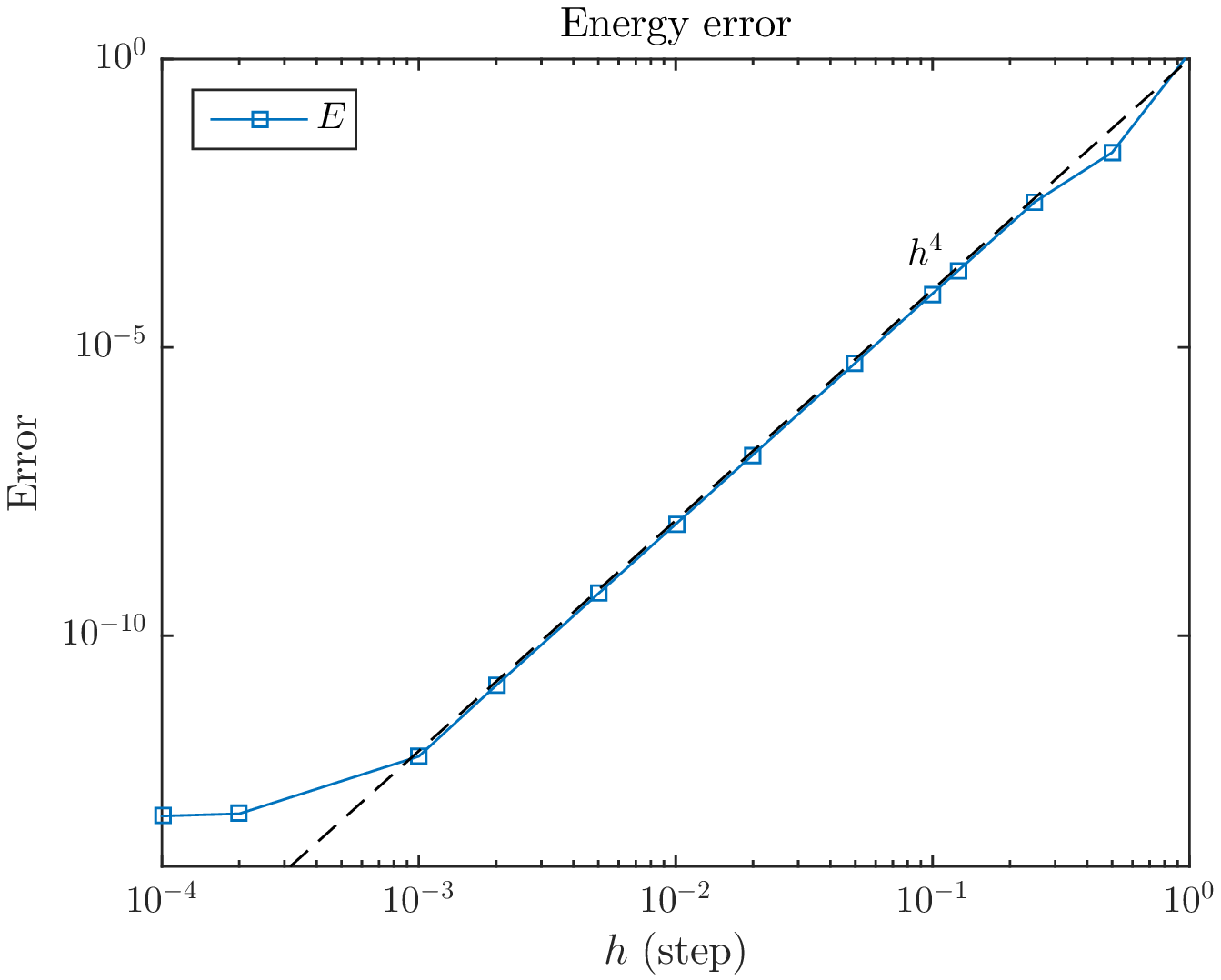}
  \end{subfigure}
  \caption{Numerical error on each separate component (left) and on the energy (right) for the Lobatto 3 method in a single simulation.}
  \label{fig:Lobatto_3}
\end{figure}

\begin{figure}[h!]
  \centering
  \begin{subfigure}[b]{0.45\textwidth}
  	\includegraphics[scale=0.45]{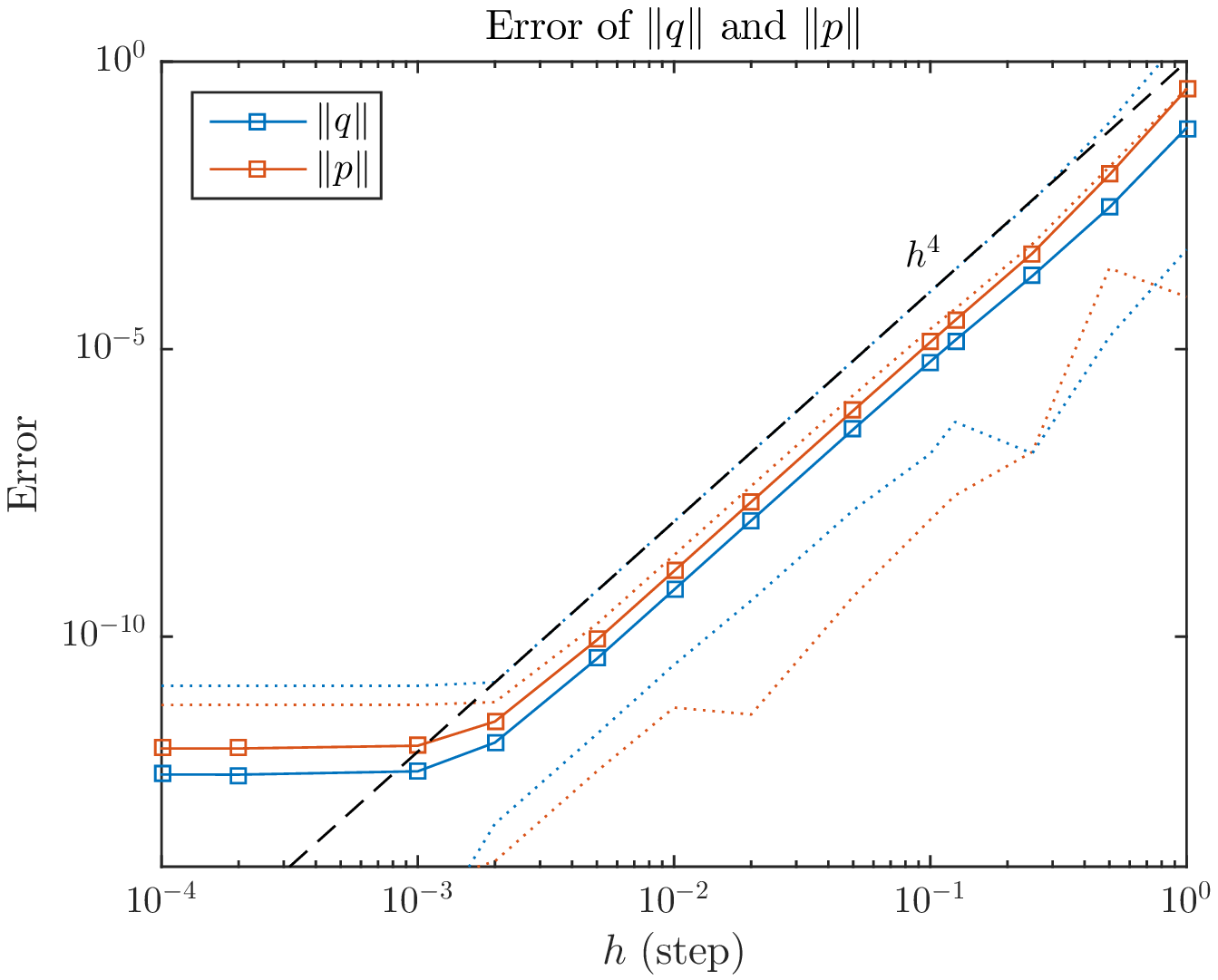}
  \end{subfigure}
  \hspace{0.5cm}
  \begin{subfigure}[b]{0.45\textwidth}
  	\includegraphics[scale=0.45]{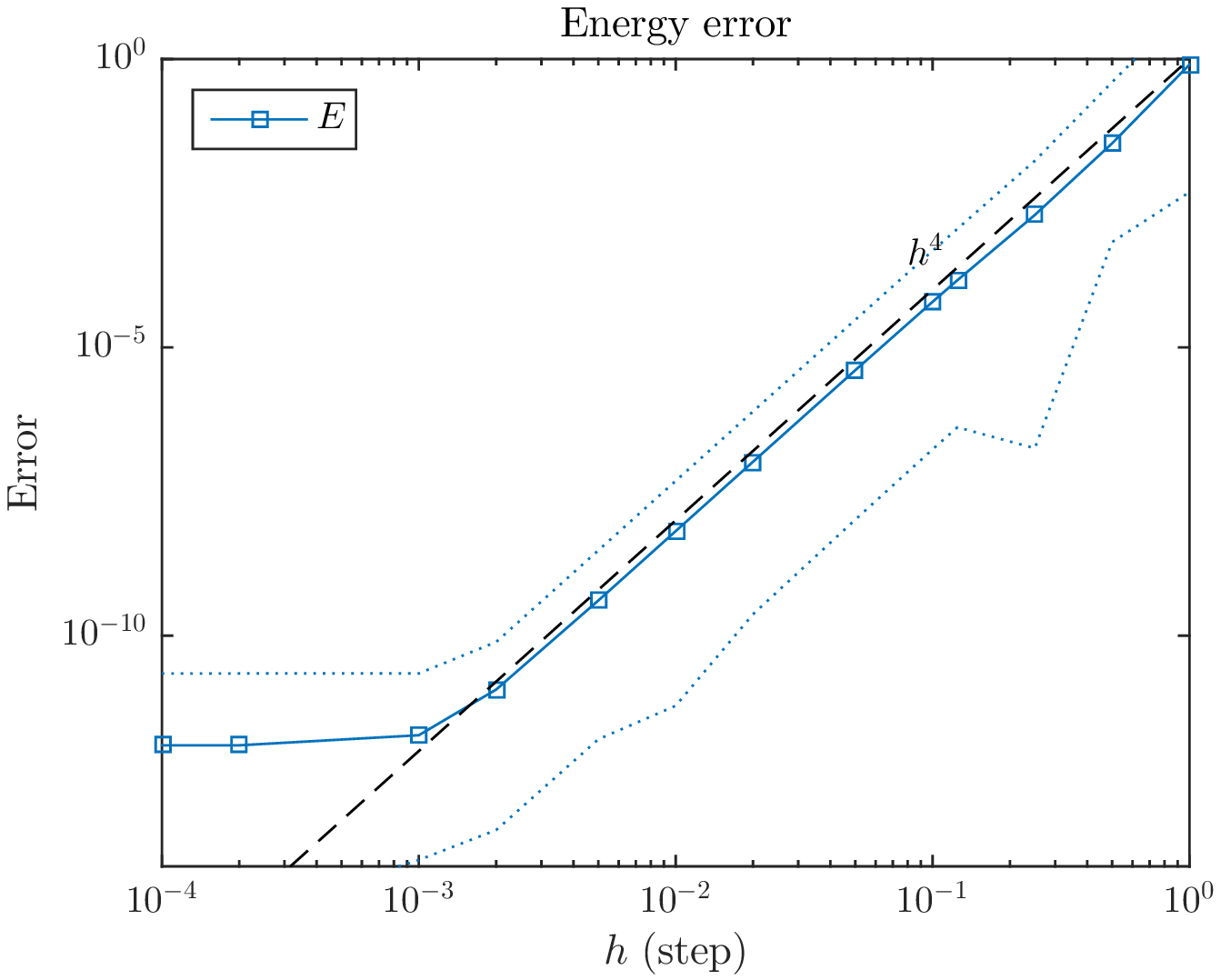}
  \end{subfigure}
  \caption{Numerical error of an ensemble for the Lobatto 3 method. Error in the norm of $q$ and $p$ (left) and on the energy (right). Dotted lines represent maximum and minimum of ensemble.}
  \label{fig:Lobatto_3_averages}
\end{figure}

\begin{figure}[h!]
  \centering
  \includegraphics[scale=0.45]{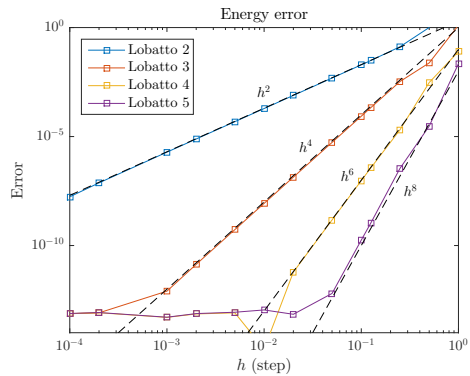}
  \caption{Errors in the energy for the different Lobatto methods.}
  \label{fig:Lobatto_composite}
\end{figure}

We chose to discretize the corresponding generalized Lagrangian, $\mathbf{L}_{\widetilde{K}_F}$, using Lobatto schemes of $2$, $3$, $4$ and $5$ stages. The order of an $s$-stage Lobatto method is $p = 2 s - 2$ so the resulting numerical methods are of order $2$, $4$, $6$ and $8$ respectively. The parameters used for the numerical simulations shown here are $(\epsilon,\rho,\lambda) = (0.5, 0.02, 0.8)$, for no particular reason. The other choices of parameters that were tested showed essentially the same behaviour. We run each simulation for a total of $1$ unit of simulation time with several different choices of step-size $h$ ranging between $5 \cdot 10^{-5}$ and $1$ and measure numerical error as the difference between the final value of the magnitude in study found for a reference simulation and the corresponding one for that we want to study. In this case our reference is taken as the simulation with the finest step-size. The initial values chosen for the results on diagrams \ref{fig:Lobatto_3} and \ref{fig:Lobatto_composite} are $(q_1, q_2, v_1, v_2) = (-1/2, -1/4, 0, 4)$. The results shown on diagram \ref{fig:Lobatto_3_averages} were found as the average from an ensemble of 25 random initial values in the square $[-4,4]\times[-4,4]$ and the pointed lines represent the maximum and minimum values found in said ensemble.

For the resolution of the resulting non-linear system of equations derived for each method, we used MATLAB's \verb|fsolve| with \verb|TolX=1e-12| and \verb|TolX=1e-14|, which explains the flat tails.

Diagram \ref{fig:Lobatto_composite} is a composite plot showing the error in the energy for the different Lobatto methods tested. The results are in agreement with the result of theorem \ref{main-theorem}. We have chosen to show only the energy to avoid clutter, but the same holds for each of the components of the system.

\section{Geometry of the method of duplication of variables}\label{sec:geometry}
In this section, we analyze the interesting geometry related with the proposed method of duplication of variables that is related with some results about symplectic groupoids. Additionally, this section will allow us in the future extend our results for reduced systems, and in general for Hamiltonian systems defined on Lie algebroids. 
We will review the definition of Lie groupoid and its associated Lie algebroid and then we will introduce the notion of symplectic Lie algebroid.

\subsection{Lie groupoids and algebroids}\label{algebroide-grupoide}

First of all, we will recall some definitions related with Lie groupoid and Lie algebroids.
 (for more details, see \cite{mackenzie87}).

\begin{definition}
A groupoid over a set $Q$ is a set $G$ together with the
following structural maps:
\begin{itemize}
\item A pair of maps $\alpha: G \to Q$, the {\sl source}, and
$\beta: G \to Q$, the {\sl target}. Thus, we can think an element $g \in G$ an arrow from $x= \alpha(g)$ to $y = \beta(g)$ in $Q$
$$
\xymatrix{*=0{\stackrel{\bullet}{\mbox{\tiny
 $x=\alpha(g)$}}}{\ar@/^1pc/@<1ex>[rrr]_g}&&&*=0{\stackrel{\bullet}{\mbox{\tiny
$y=\beta(g)$}}}}
$$
The source and target mappings define the set of composable pairs
$$
G_{2}=\{(g,h) \in G \times G / \beta(g)=\alpha(h)\}.
$$
\item A {\sl multiplication} on composable elements $\mu: G_{2} \to G$, denoted simply by $\mu(g,h)=gh$, such that
\begin{itemize}
\item $\alpha(gh)=\alpha(g)$ and $\beta(gh)=\beta(h)$.
\item $g(hk)=(gh)k$.
\end{itemize}
If $g$ is an arrow from $x = \alpha(g)$ to $y = \beta(g) =
\alpha(h)$ and $h$ is an arrow from $y$ to $z = \beta(h)$ then
$gh$ is the composite arrow from $x$ to $z$
$$\xymatrix{*=0{\stackrel{\bullet}{\mbox{\tiny
 $x=\alpha(g)=\alpha(gh)$}}}{\ar@/^2pc/@<2ex>[rrrrrr]_{gh}}{\ar@/^1pc/@<2ex>[rrr]_g}&&&*=0{\stackrel{\bullet}{\mbox{\tiny
 $y=\beta(g)=\alpha(h)$}}}{\ar@/^1pc/@<2ex>[rrr]_h}&&&*=0{\stackrel{\bullet}{\mbox{\tiny
 $z=\beta(h)=\beta(gh)$}}}}$$
\item An identity section $\epsilon: Q \to G$ of $\alpha$ and $\beta$, such that
\begin{itemize}
\item $\epsilon(\alpha(g))g=g$ and $g\epsilon(\beta(g))=g$.
\end{itemize}
\item An {\sl inversion map} $\iota: G \to G$, to be denoted simply by $\iota(g)=g^{-1}$, such that
\begin{itemize}
\item $g^{-1}g=\epsilon(\beta(g))$ and $gg^{-1}=\epsilon(\alpha(g))$.
\end{itemize}
$$\xymatrix{*=0{\stackrel{\bullet}{\mbox{\tiny
 $x=\alpha(g)=\beta(g^{-1})$}}}{\ar@/^1pc/@<2ex>[rrr]_g}&&&*=0{\stackrel{\bullet}{\mbox{\tiny
 $y=\beta(g)=\alpha(g^{-1})$}}}{\ar@/^1pc/@<2ex>[lll]_{g^{-1}}}}$$
\end{itemize}
\end{definition}
A groupoid $G$ over a set $Q$ will be denoted simply by the symbol
$G \rightrightarrows Q$.

The groupoid $G \rightrightarrows Q$ is said to be a {\sl Lie
groupoid} if $G$ and $Q$ are differentiable manifolds and all the structural maps
are differentiable with $\alpha$ and $\beta$ differentiable
submersions. If $G \rightrightarrows M$ is a Lie groupoid then $\mu$
is a submersion, $\epsilon$ is an immersion and $\iota$ is a
diffeomorphism. Moreover, if $x \in M$, $\alpha^{-1}(x)$ (resp.,
$\beta^{-1}(x)$) will be said the $\alpha$-fiber (resp.,
the $\beta$-fiber) of $x$.

Typical examples of Lie groupoids are: the pair or banal groupoid $Q \times Q$ over $Q$ (the example that we have used along all this paper), a Lie group $G$ (as a Lie groupoid over a single point), the Atiyah groupoid $(Q \times Q)/G$ (over $Q/G$) associated with a free and proper action of a Lie group $G$ on $Q$ , etc.

\begin{definition}
 If $G \rightrightarrows Q$ is a Lie groupoid and $g \in G$ then the left-translation by
$g \in G$ and the right-translation by $g$ are the
diffeomorphisms
$$
\begin{array}{lll}
l_{g}: \alpha^{-1}(\beta(g)) \longrightarrow
\alpha^{-1}(\alpha(g))&; \; \;& h \longrightarrow
l_{g}(h) = gh, \\
r_{g}: \beta^{-1}(\alpha(g)) \longrightarrow
\beta^{-1}(\beta(g))&; \; \;& h \longrightarrow r_{g}(h) = hg.
\end{array}
$$
\end{definition}
Note that $l_{g}^{-1} = l_{g^{-1}}$ and $r_{g}^{-1} = r_{g^{-1}}$.

\begin{definition}
A vector field $\xi\in {\mathfrak X}(G)$ is said to be
{\sl left-invariant} (resp., right-invariant) if it is
tangent to the fibers of $\alpha$ (resp., $\beta$) and
${\xi}(gh) = (T_{h}l_{g})({\xi}_{h})$ (resp.,
${\xi}(gh) = (T_{g}r_{h})({\xi }(g)))$, for $(g,h) \in
G_{2}$.
\end{definition}

The infinitesimal version of a Lie groupoid is a Lie algebroid which is defined as follows.

\begin{definition}
A {\sl Lie algebroid} is a real vector bundle $A\rightarrow Q$ equipped with a Lie bracket 
$\lcf \cdot ,\cdot \rcf$ on its sections $\Gamma(A)$ and a bundle map $\rho: A\rightarrow TQ$ called the {\sl anchor map} such that 
the homomorphism of $C^{\infty}(Q)$-modules induced by the anchor map, that we also denote by $\rho: \Gamma(A)\rightarrow {\mathfrak X}(Q)$, verifies 
\[
\lcf X, f Y\rcf =f\lcf X, Y\rcf +\rho(X)(f) Y,
\]
for $X, Y\in \Gamma (A)$ and $f\in C^{\infty}(Q)$.
\end{definition}

With this definition the anchor map $\rho: \Gamma(A)\rightarrow {\mathfrak X}(Q)$
is a Lie algebra homomorphism, where ${\mathfrak X}(Q)$ is endowed with the usual Lie bracket of vector field $[\cdot, \cdot]$. 

\begin{definition}
Given a Lie groupoid $G \rightrightarrows Q$, the {\sl associated Lie algebroid} $AG\rightarrow Q$ is given by its fibers 
 $A_{q}G = V_{\epsilon(q)}\alpha = Ker
(T_{\epsilon(q)}\alpha)$.
There is a bijection between the space $\Gamma (AG)$ and the set of
left-invariant vector fields on $G$. If
$X$ is a section of $\tau: AG \to Q$, the corresponding
left-invariant vector field on $G$ will
be denoted $\lvec{X}$ (resp., $\rvec{X}$), where
\begin{equation}\label{linv}
\lvec{X}(g) = (T_{\epsilon(\beta(g))}l_{g})(X(\beta(g))),
\end{equation}
for $g \in G$. Using the above facts, one may introduce a 
bracket $\lcf\cdot , \cdot\rcf$ on the space of sections $\Gamma(AG)$ and a bundle map $\rho: AG \to TQ$, which
are defined by
\begin{equation}\label{LA}
\lvec{\lcf X, Y\rcf} = [\lvec{X}, \lvec{Y}], \makebox[.3cm]{}
\rho(X)(q) = (T_{\epsilon(x)}\beta)(X(q)),
\end{equation}
for $X, Y \in \Gamma(AG)$ and $q \in Q$. 
\end{definition}
Using that $[\cdot, \cdot]$ induces a Lie algebra structure on the space of vector fields on $G$, it is easy to prove that $\lcf \cdot, \cdot \rcf$ also defines a Lie algebra structure on $\Gamma(AG)$. In addition, it follows that
\[
\lcf X, f Y\rcf = f \lcf X, Y\rcf + \rho(X)(f) Y,
\]
for $X, Y \in \Gamma(AG)$ and $f \in C^{\infty}(Q)$.

One can also stablish a bijection between sections $X\in\Gamma(AG)$ and right invariant vector fields $\rvec{X}\in {\mathfrak X}(G)$ defined by 
\begin{equation}\label{rinv}
\rvec{X}(g) = -(T_{\epsilon(\alpha(g))}r_{g})((T_{\epsilon
(\alpha(g))}\iota)( X(\alpha(g)))),
\end{equation}
which yields the Lie bracket relation
\[
\rvec{\lcf X, Y\rcf} =- [\rvec{X}, \rvec{Y}]\; .
\]

The following proposition will be useful for the results in this paper. 

\begin{proposition}\label{cft}
Let $G \rightrightarrows Q$ be a Lie groupoid and $Z\in \mathfrak{X}(G)$ a vector field invariant by the inversion, that is, 
\[
T_g \iota(Z(g))=Z(g^{-1}), \ \hbox{for all } g\in G\; .
\] 
Then, for all $q\in Q$, 
\[
Z(\epsilon(q))\in T_{\epsilon(q)}\epsilon (Q)\; .
\]
\end{proposition}
\begin{proof}
For all $v_q\in A_qG$ consider an $\alpha$-vertical curve $g: I\rightarrow G$ such that $v=\frac{\mathrm{d}g}{\mathrm{d}t}(0)$. Then
\[
T_{(\epsilon(q), \epsilon(q))}\mu(0_q, v_q)=\frac{\mathrm{d}}{\mathrm{d}t}\mu(\epsilon(q), g(t))\Big|_{t=0}=\frac{\mathrm{d}g}{\mathrm{d}t}(0)=v
\]
Also, for the $\beta$-vertical curve $g^{-1}: I\rightarrow G$ we have 
\[
T_{(\epsilon(q), \epsilon(q))}\mu(T_{\epsilon(q)}\iota(v), 0_q)=\frac{\mathrm{d}}{\mathrm{d}t}\mu(\iota(g(t)), \epsilon(q))\Big|_{t=0}=T_{\epsilon(q)}\iota(v)\; .
\]
Therefore, 
\[
T_{(\epsilon(q), \epsilon(q))}\mu(T_{\epsilon(q)}\iota(v), v)=v+T_{\epsilon(q)}\iota(v)\; .
\]
Since $\mu(g^{-1}(t), g(t))=\epsilon(\beta(g(t)))$, then
\begin{equation}\label{qwe-1}
(T_{\epsilon(q)}\iota)(v)=-v+T_{\epsilon(q)}(\epsilon\circ \beta)(v)\; .
\end{equation}
Using that
\[
Z(\epsilon(q))-T_{\epsilon(q)}(\epsilon\circ \alpha)(Z(\epsilon(q)))\in A_qG\; ,
\]
then from expression (\ref{qwe-1}):
\[
T_{\epsilon(q)}\iota(Z(\epsilon(q)))-T_{\epsilon(q)}(\epsilon\circ \alpha)(Z(\epsilon(q)))+Z(\epsilon(q))\in T_{\epsilon(q)}\epsilon(Q)\; .
\]
but from the hypothesis about $Z$, we have that
\[
T_{\epsilon(q)}\iota(Z(\epsilon(q)))=Z(\epsilon(q))\; .
\]
Therefore
\[
Z(\epsilon(q))\in T_{\epsilon(q)}\epsilon(Q)\; .
\]
\end{proof}
\subsection{Symplectic groupoid}

\begin{definition}
A {\sl symplectic groupoid} is a Lie groupoid $G \rightrightarrows Q$, such that
\begin{enumerate}
\item $(G, \omega)$ is a symplectic manifold,
\item the graph of $\mu: G_2\rightarrow G$ is a Lagrangian submanifold of $G^-\times G^-\times G$, where $G^-=(G, -\omega)$ has the negative symplectic structure.
\end{enumerate}
\end{definition}

If $G \rightrightarrows Q$ is a symplectic groupoid with symplectic form $\omega$ on $G$ then one may prove that $(\ker T_g\alpha)^{\omega}=\ker T_g\beta$, for $g\in G$, where 
\[
(\ker T_g\alpha)^{\omega}=\{v\in T_gG\, |\, \omega(v, u)=0, \hbox{for all  } u\in \ker T_g\alpha\}\; ,
\]
that is, the symplectic orthogonal of $\ker T_g\alpha$. Moreover there exists a unique Poisson structure on $Q$ such that $\alpha: G\rightarrow Q$ (respectively, 
$\beta: G\rightarrow Q$) is a Poisson (respectively, anti-Poisson) morphism. Moreover, the inversion map is an antisymplectomorphism, that is, 
$
\iota^*\omega=-\omega\; .
$

\begin{example}
{\rm
Let $G \rightrightarrows Q$ be a Lie groupoid, an let $A^*G\rightarrow Q$ be the dual vector bundle of the associated Lie algebroid $AG$. Then, the {\sl cotangent groupoid} $T^*G \rightrightarrows A^*G$ is a symplectic groupoid with the canonical symplectic form $\omega_G$. Given $\mu\in T_g^*G$, the source and target mappings are defined 
\[
\langle \tilde{\alpha}(\mu), X(\alpha(g))\rangle=\langle \mu, \rvec{X}(g)\rangle , \quad \langle \tilde{\beta}(\mu), X(\beta(g))\rangle=\langle \mu, \lvec{X}(g)\rangle\; .
\]
for all $X\in \Gamma(AG)$.
(See \cite{coste,marle1,MaMaSt} for more details and the definition of the remaining structure maps of this Lie groupoid).
}
\end{example}

\begin{proposition}
\label{prop:tangency_to_identities}
Let $G \rightrightarrows Q$ be a symplectic groupoid with symplectic form $\omega_G$ symplectic groupoid and $E: G\rightarrow {\mathbb R}$ a function such that $E\circ \iota=-E$. Then, the corresponding Hamiltonian vector field $X_E$
\[
\imath_{X_E}\omega=\mathrm{d}E\; ,
\]
verifies that $X_E(\epsilon(q))\in T_{\epsilon(q)}\epsilon (Q)$ for all $q\in Q$.
\end{proposition}
\begin{proof}
Since $\iota^*\omega=-\omega$ then for all $Y\in {\mathfrak X}(G)$
\[
\langle \mathrm{d}E, Y\rangle=\omega(X_E, Y)=-\iota^*\omega(X_E, Y)=-\omega( \iota_*X_E, \iota_*(Y))\; .
\]
but form the hypothesis we have that
\[
\langle \mathrm{d}E, Y\rangle=-\langle \mathrm{d}(E\circ \iota), Y\rangle=-\langle \mathrm{d}E, \iota_*(Y)\rangle=-\omega(X_E, \iota_*(Y))\; .
\]
Therefore, from Proposition \ref{cft} we deduce that $X_E (\epsilon (q))\in T_{\epsilon(q)}\epsilon(Q)$.
\end{proof}

\section{Conclusions and future work}

The main contributions of this paper are:
\begin{enumerate}
\item Using the duplication of variables we have rigorously deduced the error analysis of forced lagrangian systems in terms of variational error.
\item  With this technique it is possible to design efficient numerical methods for forced lagrangian systems  using previous results for variational integrators including high-order methods.
 \item We have completely elucidated the geometry of  the procedure of duplication of variables connecting with the concept of symplectic groupoid.
 \item Moreover, we have separately study the hamiltonian and lagrangian formalism and stablished the relation between both. 
 \end{enumerate}
 
 Future work includes:
\begin{enumerate}
\item Let $G \rightrightarrows M$ be a Lie groupoid with structural maps
\[
\alpha, \beta: G \to M, \; \; \epsilon: M \to G, \; \; \iota: G \to G,
\; \; \mu: G_{2} \to G.
\]
Suppose that $\tau: AG \to M$ is the Lie algebroid of $G$ and that
${\mathcal P}^{\tau}G $ is the prolongation of $G$ over the
fibration $\tau: AG \to M$, that is,
\[
{\mathcal P}^{\tau}G = AG \mbox{$\;$}_\tau \kern-3pt\times_\alpha
G \mbox{$\;$}_\beta \kern-3pt\times_\tau AG.
\]
It is clear that ${\mathcal P}^{\tau}G$ is equipped with a  Lie groupoid structure over $AG$ but also the vector bundle $\pi^\tau:{\mathcal
P}^{\tau}G \to G$ admits an integrable Lie algebroid structure (see \cite{MMM06Grupoides}). This is the corresponding version for reduced systems of the space with ``duplicated variables". We will check in a future paper how apply this methodology to analyze the order 
of geometric integrators for forced systems using  this method. 
\item An example interesting will be the case of Euler-Poincar\'e equations and double bracket dissipation \cite{BKMR}. We will study the possibility of constructing geometric integrators preserving some of the geometric structure. For instance, it can be checked that in particular examples,  the energy is dissipated but the angular momentum is not. 
\end{enumerate}

\section*{Acknowledgements}
The authors have been partially supported by Ministerio de Econom\'ia, Industria y Competitividad (MINEICO, Spain) under grants MTM 2013-42870-P, MTM 2015-64166-C2-2P, MTM2016-76702-P and ``Severo Ochoa Programme for Centres of Excellence'' in R\&D (SEV-2015-0554). R. Sato has been financially supported by a FPI scholarship from MINEICO, Spain. 

%This work has been partially supported by grants MTM 2012-34478, MTM 2013-42 870-P, MTM 2015-64166- C2-2P (MINECO), the European project IRSES-project ``Geomech-246981'' and the ICMAT Severo Ochoa projects SEV-2011-0087 and SEV-2015-0554 (MINECO). JCM acknowledges the partial support from IUMA (University of Zaragoza) for a stay at the University of Zaragoza where this work was started.

\begin{appendices}
\section[Invariant sets from discrete symmetries of a Lagrangian]{Invariant sets defined from discrete symmetries of a Lagrangian function}

The Lagrangian $\mathbf{L}_K: TQ\times TQ\rightarrow {\mathbb R}$ verifies from construction that $\mathbf{L}_K\circ \tilde{\iota}=-\mathbf{L}_K$ 
where 
 $\tilde{\iota}: TQ\times TQ\rightarrow TQ\times TQ$ is the inversion mapping given by $
\tilde{\iota}(u_q, v_{q'})=(v_{q'}, u_{q}).
$
Then the identity set $\tilde{\epsilon}(TQ)$ is invariant set for the flow of $X_{E_{\mathbf{L}_K}}$ as a consequence of the following Proposition. Observe that $\tilde{\iota}$ is the tangent lift of the map 
$Q\times Q\rightarrow Q\times Q$ given by $(q_k, q_{k+1})\mapsto (q_{k+1}, q_k)$.
\begin{proposition}\label{propo1}
Let $L: TQ\rightarrow {\mathbb R}$ be a regular Lagrangian and $\varphi: Q\rightarrow Q$ a diffeomorphism verifying that 
$L\circ \varphi_*=\pm L$. Denote by $M_\varphi=\{v_q\in Q\; |\; \varphi_*(v_q)=v_q\}$. Then $M_\varphi$ is an invariant set for any solution of the Euler-Lagrange equations. 
\end{proposition}
\begin{proof}
Consider the action sum 
\[
\begin{array}{rrcl}
 {\mathcal J}_L:& C^2(q_0, q_1, [a,b])&\longrightarrow& {\mathbb R}\\
   & c&\longmapsto& \displaystyle \int_a^b L(c(t), \dot{c}(t))\; \mathrm{d}t
\end{array}
\]
where 
$C^2(q_0, q_1, [a,b])=\{c: [a, b]\rightarrow Q\; |\; c \hbox{ is} \ C^2, c(a)=q_0, c(b)=q_1\}$.
Its tangent space is
\[
T_c C^2(q_0, q_1, [a,b])=\{ X: [0, T]\rightarrow TQ\; |\;  X \hbox{ is}\ C^1, \tau_Q(X(t))=c(t), X(a)=0, X(b)=0\} 
\]
Typically, we express $X\in T_c C^2(q_0, q_1, [a,b])$ as the tangent vector to a curve at $s=0$ in $C^2(q_0, q_1, [a,b])$,
\[
s\in (-\epsilon, \epsilon)\subset {\mathbb R}\longmapsto c_s\in  C^2(q_0, q_1, [a,b])
\]
with $c_0=c$. That is, 
\[
X=\frac{\mathrm{d} c_s}{\mathrm{d}s}\Big|_{s=0} \]

We then have
\begin{eqnarray*}
\mathrm{d}{\mathcal J}_L(c)(X)&=&\frac{\mathrm{d}}{\mathrm{d}s}\Big|_{s=0} ({\mathcal J}_L(c_s))\\
&=&\int_a^b \frac{\mathrm{d}}{\mathrm{d}s}\Big|_{s=0} L(c_s(t), \dot{c}_s(t))\; dt
\end{eqnarray*}
Using that 
$L\circ \varphi=\pm L$
\begin{eqnarray*}
\mathrm{d}{\mathcal J}_L(c)(X)&=&\pm\int_0^T \frac{\mathrm{d}}{\mathrm{d}s}\Big|_{s=0} L(\varphi(c_s(t)), \frac{\mathrm{d}}{\mathrm{d}t}(\varphi(c_s(t))))\; dt\\
&=&\pm\frac{\mathrm{d}}{\mathrm{d}s}\Big|_{s=0} ({\mathcal J}_L(\varphi\circ c_s))\\
&=&\pm \mathrm{d}{\mathcal J}_L(\varphi \circ c)(\varphi_*\circ X)
\end{eqnarray*}
Observe that $\varphi_*\circ X\in T_{\varphi\circ c} C^2(\varphi(q_0), \varphi(q_1), [a,b])$. Since $\varphi$ is a diffeomorphism then 
$c$ is a critical point of ${\mathcal J}_L$ iff $\varphi\circ c$ is a critical point of ${\mathcal J}_L$.

Now, if $c: [a, b]\rightarrow Q$ is a solution of the Euler-Lagrange equations ($d{\mathcal J}_L(c)=0$) with $\dot{c}(a)\in M_\varphi$ then also $\varphi\circ c$ is a solution of the Euler-Lagrange equations. Observe that 
\[
\frac{\mathrm{d} (\varphi\circ c)}{\mathrm{d}t}(a)=\varphi_*(\dot{c}(a))=\dot{c}(a)
\]
Then, $\varphi\circ c$ and $c$ are solutions of the Euler-Lagrange equations with the same initial conditions. Since $L$ is regular, it implies that $c=\varphi\circ c$. 
and $\dot{c}(t)\in M_\varphi$, for all $t$.

\end{proof}

\section[Invariant sets from discrete symmetries of a discrete Lagrangian]{Invariant sets defined from discrete symmetries of a discrete Lagrangian function}

The discrete Lagrangian $\mathbf{L}^d_K: TQ\times TQ\rightarrow {\mathbb R}$ verifies that $\mathbf{L}^d_K\circ \tilde{\i}_d=-\mathbf{L}^d_K$ . The following proposition gives the required result as a particular case.

\begin{proposition}\label{propo2}
Let $L_d: Q\times Q\rightarrow {\mathbb R}$ be a regular discrete Lagrangian and $\varphi_d: Q\rightarrow Q$ a diffeomorphism verifying that 
$L_d\circ (\varphi_d\times \varphi_d)=\pm L_d$. Denote by $M_{\varphi_d}=\{(q, q')\in Q\; |\; \ \varphi_d(q)=q, \varphi_d(q')=q'\}$. Then $F_{\varphi_d}$ is an invariant set for any solution of the discrete Euler-Lagrange equations. 
\end{proposition}
\begin{proof}
The proof is similar to that of Proposition \ref{propo1}. 
Consider the space
\[
{\mathcal C}_d(q_0, q_N)=\{q_d: {k}_{k=0}^N\longrightarrow Q\; |\; q_0,q_N \;\text{fixed}\}
\]
and the discrete action sum 
\[
\begin{array}{rrcl}
S_d:& {\mathcal C}_d(q_0, q_N)&\longrightarrow& {\mathbb R}\\
   & q_d&\longmapsto& \sum_{k=0}^{N-1}L_d(q_k, q_{k+1})
\end{array}
\]
The extremals are characterized as the solutions of  the discrete Euler-Lagrange equations:
\[
D_1L_d(q_k, q_{k+1})+D_2 L_d(q_{k-1}, q_k)=0, \ k=1, \ldots, N-1\; .
\]
Then, it is clear that if $\{q_{k}\}_{k=0, \ldots,  N}$ is a solution of the discrete Euler-Lagrange equations, then from the invariance of $L_d$ we easily derive that 
$\{\varphi_d(q_{k})\}_{k=0, \ldots,  N}$ is also a solution with boundary conditions $\varphi_d(q_0)$ and $\varphi_d(q_N)$. 

Therefore, if $L_d$ is regular we have defined its discrete flow  or discrete Lagrangian map: 
$$
\begin{array}{cccc}
F_{L_d}: & Q\times Q & \longrightarrow & Q \times Q \\
& (q_{k-1},q_k) & \longmapsto & (q_k,q_{k+1}) \, ,
\end{array}
$$
observe that also $F_{L_d}(\varphi_d(q_{k-1}), \varphi_d(q_{k}))=(\varphi_d(q_{k}), \varphi_d(q_{k+1})$.
Now starting from initial conditions $(q_{0}, q_1)\in M_{\varphi_d}$, that is, $\varphi_d(q_0)=q_1, \varphi_d(q_1)=q_1$ from the unicity of solutions of the discrete Euler-Lagrange equations we obtain that  $(q_{k-1}, q_k)\in M_{\varphi_d}$, $k=1, \ldots, N$ and, as a consequence, $M_{\varphi_d}$ is an invariant set of the discrete Euler-Lagrange equations.

\end{proof}
\end{appendices}

 \bibliography{References}

\end{document}